%% file: main.tex
\documentclass[reprint,amsfonts,amssymb,amsmath,showkeys,pra,superscriptaddress,twocolumn,longbibliography]{revtex4-2}
\input{preamble}

\begin{document}

\title{Sample-efficient Model-based Reinforcement Learning for Quantum Control}

\author{Irtaza Khalid}
\email{khalidmi@cardiff.ac.uk}
\affiliation{School of Computer Science and Informatics, Cardiff University, Cardiff, CF24 4AG, UK}

\author{Carrie A.\ Weidner}
\email{c.weidner@bristol.ac.uk}
\affiliation{Quantum Engineering Technology Laboratories, H.\ H.\ Wills Physics Laboratory and Department of Electrical and Electronic Engineering, University of Bristol, Bristol BS8 1FD, UK}

\author{Edmond A.\ Jonckheere}
\email{jonckhee@usc.edu}
\affiliation{Department of Electrical and Computer Engineering, University of Southern California, Los Angeles, CA 90007, USA}

\author{Sophie G.\ Schirmer}
\email{lw1660@gmail.com}
\affiliation{Department of Physics, Swansea University, Swansea, SA2 8PP, UK}

\author{Frank C.\ Langbein}
\email{frank@langbein.org}
\affiliation{School of Computer Science and Informatics, Cardiff University, Cardiff, CF24 4AG, UK}

\begin{abstract}
We propose a model-based reinforcement learning (RL) approach for noisy time-dependent gate optimization with reduced sample complexity over model-free RL. Sample complexity is defined as the number of controller interactions with the physical system. Leveraging an inductive bias, inspired by recent advances in neural ordinary differential equations (ODEs), we use an auto-differentiable ODE, parametrized by a learnable Hamiltonian ansatz, to represent the model approximating the environment, whose time-dependent part, including the control, is fully known. Control alongside Hamiltonian learning of continuous time-independent parameters is addressed through interactions with the system. We demonstrate an order of magnitude advantage in sample complexity of our method over standard model-free RL in preparing some standard unitary gates with closed and open system dynamics, in realistic computational experiments incorporating single shot measurements, arbitrary Hilbert space truncations, and uncertainty in Hamiltonian parameters. Also, the learned Hamiltonian can be leveraged by existing control methods like GRAPE for further gradient-based optimization with the controllers found by RL as initializations. Our algorithm, which we apply to nitrogen vacancy (NV) centers and transmons, is well suited for controlling partially characterized one- and two-qubit systems.
\end{abstract}

\maketitle

\section{Introduction}

Control of quantum devices for practical applications requires overcoming a unique set of challenges~\cite{koch_survey}. One is to find robust controls for noisy systems, where typical noise sources include control and feedback noise, system parameter mischaracterization, measurement and state preparation errors, decoherence and cross-talk~\cite{IBM_noises_model}. To achieve scalable, fault-tolerant quantum devices~\cite{gottesman, flammia_fault_tolerance, gambetta_fault_tolerance}, control algorithms must produce controls resilient to such noise. Reinforcement learning (RL) approaches appear more likely to find robust controls for certain applications~\cite{self1} at the cost of requiring a large number of measurements from the quantum device (samples). We propose a model-based RL approach to address this problem.

Typically, a quantum control problem is formulated as an open-loop optimization problem based on a model~\cite{GRAPE,krotov, sophie_grape, koch_survey}, which may be constructed \emph{ab initio} or obtained via a process tomography approach. During optimization there is no interaction between the physical system to be controlled and the control algorithm. The underlying assumption is that the model represents the system sufficiently accurately. This class of control algorithms has low sample complexity (high sample efficiency) represented by the number of optimization function calls until successful termination. The reason for this is, generally, that an analytical model, in particular gradient information, can be leveraged. This is a strong assumption, at least in the noisy intermediate scale quantum era, where noise impedes perfect characterization of quantum devices. However, the approach has merit, since significant thought goes into modelling and engineering quantum devices~\cite{c3}.

Alternatively, RL seeks an optimal control via interaction with the physical system, building models to various degrees. It successfully addresses challenging, noisy quantum control problems with the promise of inherent robustness~\cite{murphy, self1, self2, dalgaard2020global, phys_rev_single_shot_reward, bukov}. There are also gradient-free approaches~\cite{rabitz_neldermead} and methods that estimate gradients using variations of automatic differentiation~\cite{c3, schaefer_diff_schro, sde_autodiff, semiautodiff, schuster_lab_first_autodiff}.

RL approaches utilizing only measurements without prior information do not suffer from model bias. Moreover, they usually optimize an average controller performance over the noise in the system, yielding inherently robust controllers~\cite{self2}. However, this means the number of optimization function calls becomes prohibitively large, and RL's high sample complexity is a core problem limiting its practical applicability~\cite{barto_rlbook}. This is not surprising as without a prior model little or no information is available to the optimization algorithm and all information must be obtained via measurements.

Despite this inherent restriction, in recent times, RL has been deployed on real quantum devices for parametrized pulse-level gate optimization~\cite{baum2021}, improving the performance of quantum error correcting codes~\cite{sivak2023} and fluxonium gate parameter optimization~\cite{ding2023}. In line with forthcoming analysis, the sample complexity of these RL experiments is estimated to be around $10^4$, {$10^3$} and $10^4$, respectively, excluding the cost of estimating observables using single-shot measurements. These costs are smaller than direct or ab initio applications of RL, which consume around $10^6$ samples~\cite{barto_rlbook}, as the aforementioned works, to differing extent, exploit specific knowledge of the quantum system to frame the problem to be easier to optimize for the RL agent. More specifically, these works use custom RL adaptations for each problem, e.g., fine-tuning solutions already found by other optimization algorithms as the final step during control preparation in Ref.~\cite{ding2023}, or exploiting some experimental structure that simplifies finding optimal controls in Ref.~\cite{baum2021}. In the present paper, we remain generic in our ignorance of the system Hamiltonian during acquisition of optimal controls to demonstrate the general utility of our approach without inducing constraining (and potentially incorrect if not confidently known) biases on the learning problem. We note, however, that there is significant scope for sample-efficiency reductions. For example, the use of our model-based RL algorithm would make RL, in general, extensible to a wider class of quantum control experiments.

In classical RL, high sample complexity is typically addressed using model-based methods, which construct a model from scratch using information obtained from measurements. Such methods result in reduced sample complexity for benchmark problems~\cite{Dyna}. They are successful if the model and the measurements (samples) obtained during training possess some generalizability~\cite{PETS, MBPO} that is captured by a function approximator (usually a neural network). However, methods involving universal function approximation of dynamic trajectories are unstable. This is because learning can be hindered by the very large space of trajectories, and interpolating from insufficient sample trajectories can be shallow or incorrect~\cite{deepmind2018_when_to_use_mbrl}. More importantly, for quantum data, it is known that a time-independent Hamiltonian can generate many unitary propagators, so estimating the model may imply learning the entire Hilbert space of propagators for a particular control problem which is often intractable. This motivates learning the dynamical generator, i.e., the Hamiltonian, instead of the propagators.

In this paper, we propose a model-based RL method for time-dependent, noisy gate preparation where the model is given by an ordinary differential equation (ODE), differentiable with respect to model parameters~\cite{neural_ode}. ODE trajectories do not intersect~\cite{coddington1955theory, dupont_augmentednode, intrinsic_robustness_ode}, which constrains the space of potential models for learning and makes learning robust to noise. We parameterize the Hamiltonian by known time-dependent controls and unknown time-independent (system) parameters, which, in addition, makes the model interpretable.

We show that combining the inductive bias from this ODE model with partially correct knowledge (assuming the controls are known but not the time-independent system Hamiltonian) reduces the sample complexity compared to model-free RL by at least an order of magnitude.

It has recently been shown that inductive biases, i.e., encoding the symmetries of the problem into the architecture of the model space, such as the translation equivariance of images in the convolution operation~\cite{geometric_deep_learning}, leads to stronger out-of-distribution generalization by the learned model. This is because inductive biases impose strong priors on the space of models such that training involves exploring a smaller subset of the space to find an approximately correct model.

We demonstrate improvement over the sample-efficient soft-actor critic (SAC) model-free RL algorithm~\cite{haarnoja2018soft} for performing noisy gate control in leading quantum computing architectures: nitrogen vacancy (NV) centers (one and two qubits)~\cite{nv_center}, and transmons (two qubits)~\cite{effective_ham_gambetta}, subject to dissipation and single-shot measurement noise. We also show that the learned Hamiltonian can be leveraged to optimize the controllers found by our RL method further using GRAPE~\cite{GRAPE, sophie_grape}.

% New para noting the connection with the model predictive QHL paper https://ieeexplore.ieee.org/stamp/stamp.jsp?tp=&arnumber=9779927&tag=1
Our approach is similar in spirit to Ref.~\cite{clouatre2022mpc} where a novel Hamiltonian learning protocol via quantum process tomography is proposed for the purpose of model-predictive control. The complete Hamiltonian (including the control and system parts) is identified term by term via a Zero-Order Hold (ZOH) method, where only one term is turned on at a time, e.g., by setting the control parameters to zero, and learned individually using optimization over the Stiefel manifold. As a side remark, a sample complexity advantage between learning the Hamiltonian with quantum control than without it has recently been shown~\cite{dutkiewicz2023advantage}. The learned Hamiltonian is then used to obtain a viable control sequence for a variety of state and gate preparation problems for closed (unitary) systems under the influence of initial state preparation errors. While it is possible for our Hamiltonian learning protocol to also learn the full Hamiltonian using the ZOH method, we focus on the problem of improving the sample complexity of RL in this paper through the incorporation of a partially known physics-inspired model. Furthermore, our focus is also directed on the interplay of concurrently learning the model and controlling the system in noisy closed and open system settings.

This paper is organized as follows: in Sec.~\ref{sec:control_prob} we define the open and closed system control problems including our setup to simulate single-shot measurements and the RL control framework; Sec.~\ref{sec:mbrl_control_setup} describes the model-based version of the RL control framework and Sec.~\ref{sec:results} presents numerical studies for some realistic example control problems on the system architectures described above in noisy and ideal settings and how to leverage the learned system Hamiltonian using GRAPE.

\section{The Quantum Control Problem}\label{sec:control_prob}

We briefly introduce the quantum control problem for open and closed quantum systems and describe how we estimate the propagators from measurements, needed for our RL approach.

\subsection{Closed System Dynamics}\label{ssec:closed-system-dynamics}

Consider a quantum system that is represented by an effective Hamiltonian $H(t)$ in the space of complex Hermitian $n\times n$ matrices
\begin{equation}
  H(\mathbf{u}(t), t) = H_0 + H_c(\mathbf{u}(t), t),
  \label{eq:general_hamiltonian}
\end{equation}
where $H_0$ is the time-independent system Hamiltonian and $H_c$ is the control Hamiltonian parametrized by time-dependent controls $\mathbf{u}(t)$. Its closed-system dynamics are governed by the Schr\"odinger equation,
\begin{equation}
  \dv{U(\mathbf{u}(t), t)}{t} = -\frac{i}{\hbar}H(\mathbf{u}(t), t)U(\mathbf{u}(t), t), \quad U(t=0) = \eye,
  \label{eq:unitary_ode}
\end{equation}
where $U(\mathbf{u}(t), t)$ is the unitary propagator representing the state evolution. Its fidelity to realize a target gate $U_\text{target}$ is
\begin{equation}
  F(U_\text{target}, U(\mathbf{u}(t), t)) = \frac{1}{n^2}\left|\Tr[U_\text{target}^\dagger U(\mathbf{u}(t), t)]\right|^2.
  \label{eq:fidelity}
\end{equation}
The control problem to implement $U_\text{target}$ is
\begin{equation}
  \mathbf{u}^*(t^*) = \argmax_{\mathbf{u}(t),~t \leq T} F(U_\text{target}, U(\mathbf{u}(t), t)),
  \label{eq:control_problem_standard}
\end{equation}
where $\mathbf{u}^*(t^*)$ are the optimized control parameters for an optimized final time $t^* \leq T$.

\subsection{Open System Dynamics}\label{ssec:open_system_dynamics}

For open system dynamics consider an arbitrary state with density matrix $\rho$ for $\log_d{n}$ qudits evolving according to the master equation~\cite{open_quantum_systems_textbook,floether2012robust}
\begin{equation}
  \dv{\rho(t)}{t} = -\frac{i}{\hbar}[H(\mathbf{u}(t), t), \rho] + \mathfrak{L}(\rho(t)),
  \label{eq:master_eqn}
\end{equation}
where $\mathfrak{L}(t)$ describes the Markovian decoherence and dephasing dynamics (i.e., the environment),
\begin{equation}
   \mathfrak{L}(\rho(t)) = \sum_{d}{\gamma_{d}\left(l_{d} \rho l_{d}^\dagger - \frac{1}{2} \{ l_{d}^\dagger l_{d}, \rho\}\right)},
    \label{eq:Lindblad_operators}
\end{equation}
and $l_{d}$ is a decoherence operator that can be non-unitary.

To characterize the gate implemented by $\mathbf{u}(t)$, we need to consider the evolution of a complete orthonormal basis of states, $\{\rho_k\}_{k=1}^{n^2}$. For this we introduce the Liouville superoperator matrix $\mathbf{X}$ that acts on an arbitrary vectorized state $\boldsymbol\rho$ (e.g., obtained by stacking the matrix columns) to produce the evolution
\begin{equation}
  \boldsymbol\rho(t) = \mathbf{X}(t)\boldsymbol\rho(t=0).
  \label{eq:matrix_liouville_superoperator}
\end{equation}
This is equivalent to the tensor-matrix evolution~\cite{wood2011tensor}
\begin{equation}
  \rho(t)_{mn} = \sum_{\mu, \nu} X_{nm, \nu \mu}(t) \rho_{\mu \nu}(t=0).
  \label{eq:tensor_liouville_superoperator}
\end{equation}
$X_{nm, \nu \mu}(t)$ is a fourth order tensor (used to refer to multi-dimensional arrays in this context) form of $\mathbf{X}(t)$ that encodes the evolution of the state element $\rho_{\mu \nu}$.

Thus, similar to Eq.~\eqref{eq:unitary_ode}, we define a superoperator $X(\mathbf{u}(t), t)$ which encodes the evolution of $\{\rho_k\}_{k=1}^{n^2}$ and follows the linear ODE
\begin{equation}
  \dv{\mathbf{X}(\mathbf{u}(t), t)}{t} = -\frac{i}{\hbar}(\mathbf{L}_0 + i\mathbf{L}_1)\mathbf{X}(\mathbf{u}(t), t), \quad \mathbf{X}(t=0) = \eye
  \label{eq:superoperator_ode}
\end{equation}
where $\mathbf{L}_0, \mathbf{L}_1$ represent the superoperator version of the commutator map $[H(\mathbf{u}(t), t), \cdot]$ and $\mathfrak{L}(\cdot)$ the Markovian decoherence and dephasing dynamics.

We factorize out an imaginary prefactor $i$ to the left in Eq.~\eqref{eq:superoperator_ode} to unify the ODE for open and closed system dynamics. For $\mathfrak{L} \equiv \mathbf{0}$, the above reduces to the closed system dynamics of Eq.~\eqref{eq:unitary_ode}. For open dynamics, to be faithful to experimental limitations, we implement single-shot noise when estimating the gate, i.e., process tomography. We transform the superoperator $X_{nm, \nu \mu}$ to the Choi matrix $\Phi/\Tr[\Phi]$ that is given by index reshuffling or partial transpose (and more formally a contravariant-covariant change of coordinates)~\cite{wood2011tensor, lichnerowicz2016elements},
\begin{equation}
  \Phi_{n m,\mu \nu } = X_{\nu m,\mu n}.
  \label{eq:super_to_choi}
\end{equation}

In Sec.~\ref{sec:results}, we use this for open and closed dynamics. Estimating $\Phi$ is possible using ancilla-assisted quantum process tomography (AAPT) and the Choi-Jamiolkowski isomorphism~\cite{choi, jam, ancilla_assisted_quantum_process_tomography} for $2\log_d{n}$-qudit states and $\log_d{n}$-qudit gates. Analogously to the above, $\Phi$ has a matrix version $\choi$. In this paper, we decompose $\choi$ over a generalized $\su(n^2)$'s algebra basis $\{P_k\}_{k=1}^{n^4-1}$, e.g., Gell-Mann matrices~\cite{gellmann},
\begin{equation}
  \frac{\choi}{\Tr[\choi]} = \frac{\mathds{1}}{n^2}+ \sum_{k=2}^{n^4-1} q_k P_k
\end{equation}
whose coefficients are
\begin{equation}
  q_k = \frac{\Tr[P_k \choi]}{\Tr[\choi]} \in [-1,1].
  \label{eq:binomial_gate_measurements}
\end{equation}
$q_{k}$ can be modelled as a binomial random variable $\text{Bin}(M, p_k)$ with probability $p_{k} = \frac{1}{2}(1+q_{k})$ where $M$ is the number of single-shot (Bernoulli) measurements~\cite{optimal_control_with_poor_statistics}.
{
The Gell-Mann matrices are a generalization of the Pauli matrices and the corresponding physical measurement operations are akin to measuring qudit energy levels in an informationally complete basis.}

We measure the faithfulness of the implemented gate $\choi(\mathbf{u}(t), t)$ w.r.t. the target gate (as another Choi state) $\choi_\text{target}$ using the generalized state-fidelity~\cite{flammia_direct_fidelity_2011},
\begin{align}
  F(\choi(\mathbf{u}(t), t), \choi_\text{target})
  &= \frac{\Tr[\choi(\mathbf{u}(t), t)\choi_\text{target}]}{\Tr[\choi(\mathbf{u}(t), t)]\Tr[\choi_{\text{target}}]} \nonumber\\
  &= \frac{1}{n^4} + \sum_{k=2}^{n^4-1} q_k^\text{target} q_k.
  \label{eq:generalised_state_fidelity}
\end{align}
Analogously to the closed case, the open control problem is to find an optimal control $\mathbf{u}^*(t^*)$ for an optimal final time $t^* \leq T$ (with $T$ being the fixed upper bound), such that
\begin{equation}
  \mathbf{u}^*(t^*) = \argmax_{\mathbf{u}(t),~t \leq T} F(\choi(\mathbf{u}(t), t), \choi_\text{target}).
  \label{eq:noisy_gate_control_problem}
\end{equation}

\subsection{Discretization}\label{ssec:discretization}

The exact solution of the time-dependent general dynamics discussed in Eq.~\eqref{eq:noisy_gate_control_problem} is given by the time-ordered operator
\[
\mathbf{E}(t^*, \mathbf{u}^*(t^*)) = \mathcal{T}\exp\left(\int_{0}^{t^*}dt'\;-\frac{i}{\hbar}\mathbf{G}(t', \mathbf{u}^*(t'))\right)
\]
for a unitary or Lindbladian generator $\mathbf{G}$. In practice, we solve for a piece-wise constant version of the dynamics represented by $N$ fixed steps of $\Delta t = T/N$ of the final time $T$. Thus, $\mathbf{E}(\mathbf{u}(t), t)$ is discretized, which amounts to fixing $\mathbf{u}(t)=\mathbf{u}_m$ to be constant for each timestep such that $\mathbf{u}_m \in \mathbb{C}^{m \times C}$ is a finite dimensional array where $C$ is the number of controls per timestep in the vector $u_l$ parametrizing $H_c(u_l, t_l)$ and $m$ is the number of total timesteps in the pulse, with $m \leq N$ for a maximum number of pulse segments $N$. The propagator is
\begin{equation}
  \mathbf{E}(t, \mathbf{u}(t)) := \mathbf{E}(\mathbf{u}_m) = \prod_{l=1}^{m} \exp(-\frac{i}{\hbar}\Delta t \mathbf{G}(t_l, \mathbf{u}(t_l))).
\end{equation}
The control problems in Eqs.~\eqref{eq:control_problem_standard} and~\eqref{eq:noisy_gate_control_problem} are equivalent to
\begin{equation}
  \mathbf{u}^*_m = \argmax_{\mathbf{u}_m=[u_1, \dots, u_m] \in \mathbb{X}, m \leq N} \mathcal{F}(\choi(\mathbf{E}(\mathbf{u}_m)), \choi(\mathbf{E}_\text{target}))
  \label{eq:discretized_control_problem}
\end{equation}
for a fidelity $\mathcal{F}$ and the time. $\mathbf{u}_m$ is constrained to some maximum and minimum values given by $\mathbb{X} = \{\mathbf{u}_m: \forall c,l \; u_\text{min} \leq u_{cl} \leq u_\text{max} \in \mathbb{C}\}$. The constraints are applied separately to the real and imaginary parts of the components of $\mathbf{u}_m$.

\section{Model-based Reinforcement Learning Control}\label{sec:mbrl_control_setup}

We give a brief overview of RL, followed by explaining our model-based RL approach. An excellent introduction can be found in Ref.~\cite{barto_rlbook}.

\subsection{Reinforcement Learning for Quantum Control}

The RL problem is usually treated as a sequential Markov decision problem (MDP) on the space of states, actions, transition probabilities and rewards: $(\mathcal{S}, \mathcal{A}, \mathcal{P}, \mathcal{R})$. This describes an environment for consecutive one-step transitions, indexed by $k=1,2,\dotsc$, from current state $\s_k \in \mathcal{S}$ to next state $\s_{k+1} \in \mathcal{S}$ if an RL agent executes action $\action_k \in \mathcal{A}$, yielding immediate scalar reward $\rew_{k} \in \mathcal{R}$. The environment is generally probabilistic, so $\mathcal{P}(\s_{k+1} | \s_k, \action_k)$ is the probability that the agent is in state $\s_{k+1}$ after executing $\action_k$ in state $\s_k$. An RL agent follows a policy function that is represented by a conditional probability distribution $\pi(\action_k|\s_k)$: the probability of taking action $\action_k$ after observing the state $\s_k$.

The quantum control problem can be represented as an RL problem by sequentially constructing the control amplitudes as actions, using the unitary propagator the control implements as the state with the reward as the fidelity:
\begin{subequations}
  \begin{align}
    \action_k &= u_{k},\\
    \s_k &= \prod_{l=1}^{k} \exp(-\frac{i}{\hbar}\Delta t \mathbf{G}(t_l, u_l)),\\
    \rew_k
    &= \mathcal{F}(\choi(\mathbf{E}(\mathbf{u}_k)), \choi(\mathbf{E}_\text{target}))
    % &= \mathcal{F}(\mathbf{E}(\left[u_1,\dotsc,u_k\right]), \mathbf{E}_\text{target})
    .
  \end{align}
  \label{eq:qc_problem_for_rl}
\end{subequations}
As this is deterministic the probabilities $\mathcal{P}$ are trivial, and we have a simple environment function $\mathcal{E}: \mathcal{S} \times \mathcal{A} \rightarrow \mathcal{S} \times \mathcal{R}$, mapping the current state and action $(s,a)$ to the next state and reward $(s',r)$. In model-free RL (see Algorithm~\ref{algorithm:rl_loop}), a discounted sum of expected rewards, called the returns,
\begin{align}
  \eta(\pi) := \mathds{E}_{\action_t \sim \pi}\left[\sum_{k=0}^{\infty} \gamma^k \rew_k \right]
\end{align}
is maximized, where $\mathds{E}_{x \sim P}[\cdot] = \int_{\mathcal{X}}dx\;P(x)[\cdot]$ is the expectation operator and $0 \leq \gamma \leq 1$ is a discount factor.

\begin{algorithm}[t]
  \SetEndCharOfAlgoLine{}
  Initialize empty dataset $\mathcal{D}$, parametrized random policy $\pi_\theta$, $k \gets 0$ \;
  % final time $T$, Target gate $\mathbf{E}_\text{target}$, timestep $\Delta t$  \;
  Observe initial state $s_0$\;
  \While{$k < T/\Delta t$}{
    Execute $\action_k \gets \pi_{\nntheta}\left( \cdot | \s_k\right)$ \;
    Observe $\s_{k+1}$, $\rew_k \gets \mathcal{E}(\s_k, \action_k)$ \;
    Store $\mathcal{D} \gets \mathcal{D} \cup \{(\s_{k}, \s_{k+1}, \action_{k}, \rew_{k})\}$ \;
    $k \gets k + 1$ \;
  }
  \tcp*[l]{if require update: perform model-free update of parameters (e.g. policy $\pi_{\nntheta}$ )}
  \caption{Reinforcement learning loop}\label{algorithm:rl_loop}
\end{algorithm}

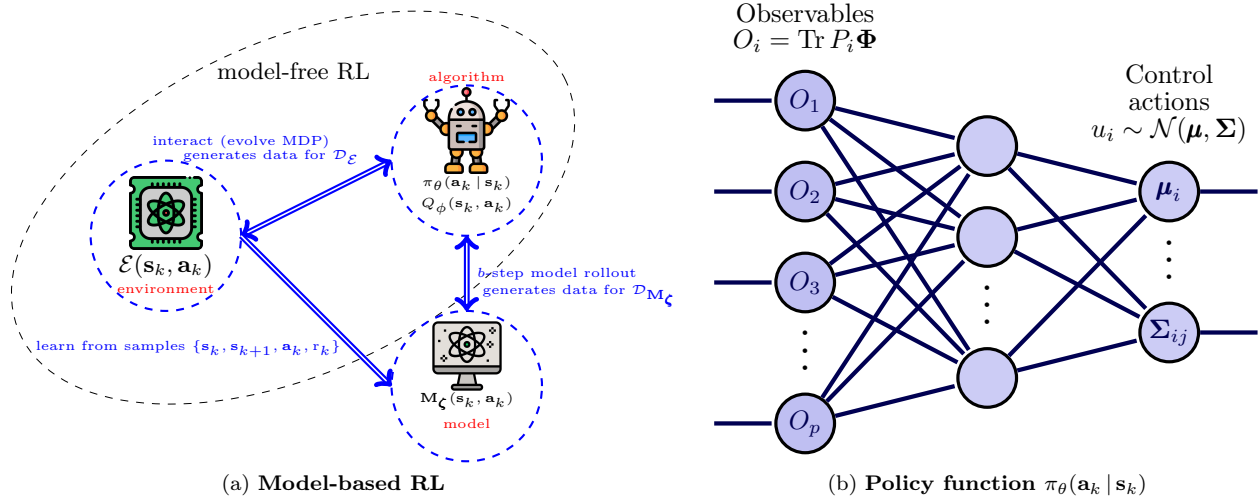
\begin{figure*}[t]
\subfloat[\textbf{Model-based RL}\label{fig:mb_rl}]
  {\begin{minipage}[t]{0.55\linewidth}
    \begin{tikzpicture}
      % model free region
      \draw[line width=0.01mm, black, dashed, rotate=25] (0,-4) ellipse (4.cm and 2.2cm) node {};
      \node[black] at (1.7,-1.8) {model-free RL};
      % environment
      \node[white] at (0,-3.7) {\includegraphics[width=3em]{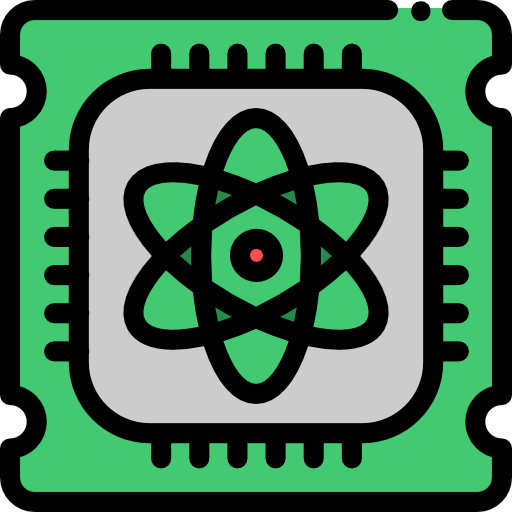}};
      \draw[blue,thick,dashed] (0,-4) circle (1cm) node {};
      \node[black] at (0,-4.4) {$\mathcal{E}(\s_k, \action_k)$};
      \node[red] at (0,-4.7) {\tiny environment};
      % interaction arrows
      \draw[blue,thick, double,<->] (1,-4) -- (3,-3);
      \draw[blue,thick, double,->] (1,-4) -- (3,-6);
      \draw[blue,thick, double,<->] (4,-4) -- (4,-5);
      \node[blue] at (0.3, -5.5) {{\tiny learn from samples $\{\s_k,\s_{k+1},\action_k,\rew_k\}$}};
      \node[blue] at (1., -2.75) {\tiny interact (evolve MDP)};
      \node[blue] at (1.4, -2.95) {\tiny generates data for $\mathcal{D}_{\mathcal{E}}$};
      \node[blue] (rollout) at (5.2, -4.5) {\tiny $b$-step model rollout};
      \node[blue] (rollout) at (5.5, -4.8) {\tiny generates data for $\mathcal{D}_{\model}$};
      % agent
      \node[white] at (4,-2.6) {\includegraphics[width=3.5em]{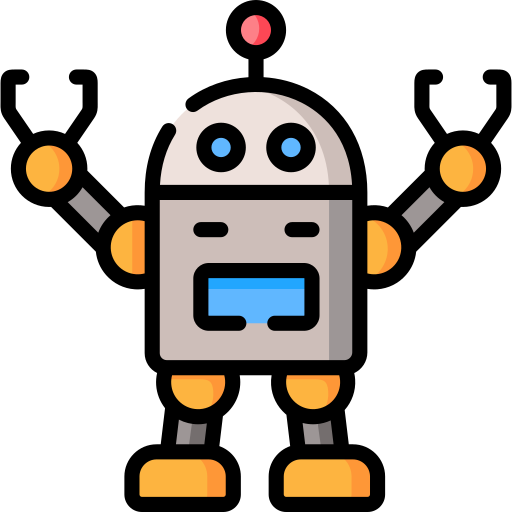}};
      \node[blue] at (4,-3.4) {};
      \draw[blue,thick,dashed] (4,-3) circle (1cm)  node {};
      \node[red] at (4, -1.9) {\tiny algorithm};
      \node[black] at (4,-3.3) {\tiny $\pi_\theta(\action_k | \s_k)$};
      \node[black] at (4,-3.6) {\tiny $Q_\phi(\s_k,\action_k)$};
      % model
      \draw[blue,thick,dashed] (4,-6) circle (1cm) node {};
      \node[black] at (4,-6.2) {\tiny $\model(\s_{k}, \action_k)$};
      \node[red] at (4,-6.5) {\tiny model};
      \node[white] at (4,-5.6) {\includegraphics[width=3em]{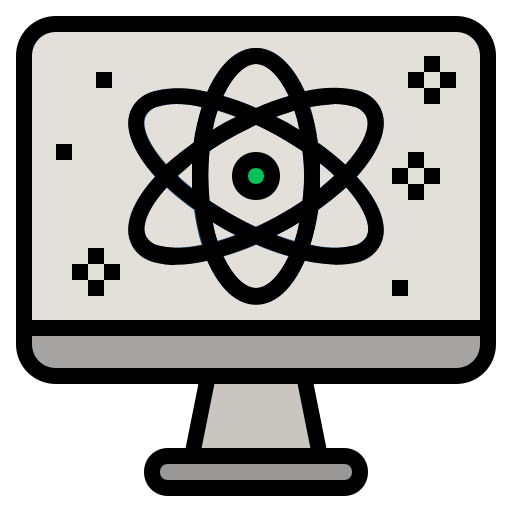}};
    \end{tikzpicture}
   \end{minipage}}
  \subfloat[\textbf{Policy function $\pi_\theta(\action_k | \s_k)$}\label{fig:policy_function_nn}]{\includegraphics[width=0.42\linewidth]{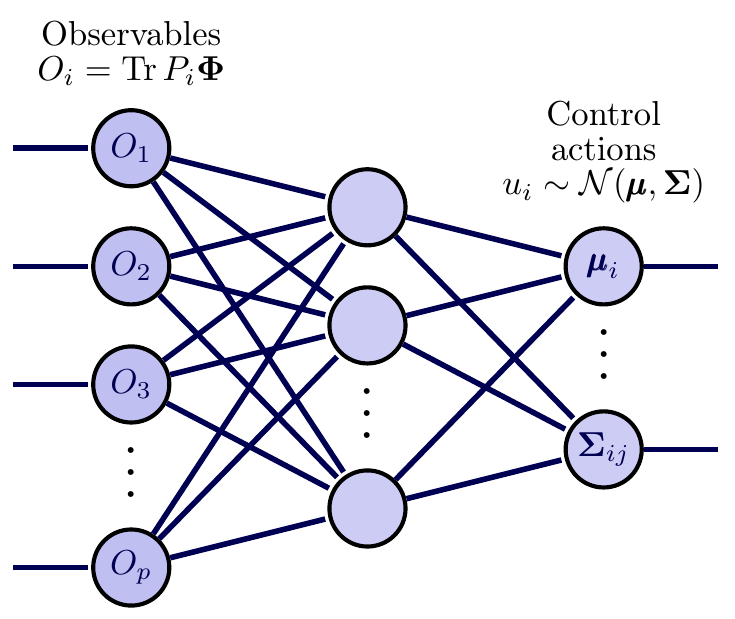}}
\caption{A schematic of model-based RL is given in~(a). The arrow-head implies direction of affect of the edge between a source and a sink node. The agent or policy function $\pi_\theta$ interacts with the RL environment modelled as MDP to collect data $\{\s_k,\s_{k+1},\action_k,\rew_k\}$. This encompasses model-free RL. The data is then used to train the model $\model(\s_{k}, \action_k)$. The model is trained until some quality measure like the validation prediction error on some untrained-upon data from the environment plateaus indicating that the training is complete. Then, it is used to generate synthetic data through a $b$-step rollout in which the policy interacts with the model $b$ times. The policy parameters $\theta$ (and the state-action value function parameters $\phi$) are optimized using the real and model generated data. In~(b), we visualize the policy inputs as the gate-characterizing observables  (unitary or Lindblad) about the Choi matrix $\choi$ given by Eq.~\eqref{eq:binomial_gate_measurements} and the tunable outputs are the parameters of a multivariate Gaussian distribution, i.e., the mean $\pmb{\mu}$ and covariance $\mathbf{\Sigma}$. The controls $u_i$ are drawn from $\mathcal{N}(\pmb{\mu},\pmb{\Sigma})$.}\label{fig:algo_schematic}
\end{figure*}

However, Refs.~\cite{haarnoja2018soft,ziebart2008maximum} observe that adding an entropy maximizing term for the policy $\pi(\action_k | \s_k)$ to the optimization objective encourages exploration of the state space $\mathcal{S}$, improves the learning rate of the agent and reduces the relative number of samples needed, compared to other standard RL algorithms. The maximum entropy objective
{
or the entropy-regularized cumulative reward function $J$
}
for $N$ steps is
\begin{equation}
  J(\pi)
  = \sum_{k=0}^N \gamma^k \mathds{E}_{(\s_k, \rew_k) \sim \mathcal{E}_\pi}\left[ r_k + \alpha J_1(\s_k) \right]
  \label{eq:rl_objective}
\end{equation}
where $\mathcal{E}_\pi$ represents the environment's state-action probability distribution induced by the policy $\pi$, $\alpha$ is an optimizable temperature parameter (signifying the importance of exploration in the objective), and $J_1(\s_k)$ is the entropy of the policy function $\pi(\cdot | \s_k)$ conditional on the $k$th state $\s_k$,
\begin{equation}
  J_1(\s_k) = -\mathds{E}_{x \sim \pi(\cdot|\s_k)}\left[\log(\pi(x|\s_k))\right].
  \label{eq:entropy}
\end{equation}
Thus, the RL control problem becomes a problem of finding the optimal control policy $\pi^*$ given by
\begin{equation}
  \pi^* = \argmax_\pi J(\pi).
  \label{eq:rl_opt_problem}
\end{equation}
This is exactly solvable for tabular MDPs using dynamic programming and heuristically with neural network function approximation for continuous MDPs.

\subsection{Model-Based Reinforcement Learning}

\begin{algorithm*}
  \SetEndCharOfAlgoLine{}
  \SetKwComment{Comment}{$\triangleright$ }{}
  \SetKwInOut{Input}{Input}
  \SetKwInOut{Output}{Output}
  \Input{\\\hspace{-3.6em}
    \begin{tabular}[t]{l @{\hspace{1em}} l}
      $H_c$ & control Hamiltonian (time-dependent part of $H(t)$ in Eq.~\eqref{eq:general_hamiltonian})\\
      $T, \Delta t, M$ &  max time, timestep size, number of single shot measurements (if open system to estimate $\pmb{\Phi}$ using Eq.~\eqref{eq:binomial_gate_measurements})\\
      $\mathbf{E}_\text{target}$ & target gate\\
      $W, C, b, \texttt{tol}$ & Epochs, timesteps, rollout length, validation loss tolerance (which is a problem-specific hyperparameter)
    \end{tabular}
  }
  \Output{\\\hspace{-3.6em}
    \begin{tabular}[t]{l @{\hspace{1em}} l}
      $\mathbf{u}^*$ & Approximately optimal 2D array of controls that solves Eq.~\eqref{eq:discretized_control_problem}\\
      $\nntheta,~\phi,~\pmb{\zeta}$ & Optimized parameters of the policy, critic and learned model\\
    \end{tabular}
  }
  Initialize empty environment dataset $\mathcal{D}_{\mathcal{E}}$, model dataset $\mathcal{D}_{\mathbf{M}_{\pmb{\zeta}}}$, random policy $\pi_\theta$ \;
  \tcp*[l]{collect random model training data}
  Populate $\mathcal{D}_{\mathcal{E}}$ using
  {
  uniform
  }
  random policy $\pi_{\nntheta}$ with Algorithm~\ref{algorithm:rl_loop} without updates \Comment*[r]{randomly explore the environment $\mathcal{E}$ state space}
  \For(){$W$ epochs}{
    \tcp*[l]{Train model}
    Sample a batch of training and validation data $D_\text{train}, D_\text{val} \sim \mathcal{D}_{\mathcal{E}}$ and minimize $L_\text{model}(D_\text{train})$ in Eq.~\eqref{eq:regression_loss} \;
    \For(){$C$ timesteps}{
      \tcp*[l]{agent-environment interaction}
      Execute $\action_k \gets \pi_{\nntheta}(\cdot | \s_k)$, observe $\s_{k+1}, \rew_k \gets \mathcal{E}(\s_k,\action_k)$ and store data $\mathcal{D}_{\mathcal{E}} \cup \{(\s_{k}, \s_{k+1}, \action_{k}, \rew_{k})\}$ \;
      \If{$L_\text{model}(D_\text{val}) < \texttt{tol}$}{
        \tcp*[l]{agent-model interaction}
        Sample uniformly a batch of initial states $\{\s_k\} \sim \mathcal{D}_{\mathcal{E}}$, $k \gets 0$\;
        \For(){$k'$ in $\{1, \cdots, b\}$}{
          \BlankLine
          Execute $\action_{k'} \gets \pi_{\nntheta}(\cdot | \s_{k'})$ and observe $\s_{k'+1}, \rew_{k'} \gets \mathbf{M}_{\pmb{\zeta}}(\s_{k'},\action_{k'})$ \Comment*[r]{$b$-length model rollout}
          Store $\mathcal{D}_{\model} \gets \mathcal{D}_{\model} \cup \{(\s_{k'}, \s_{k'+1}, \action_{k'}, \rew_{k'})\}$ \;
          $k' \gets k' + 1$,
        }
      }
      Train policy by minimizing $J'(\pi_{\nntheta})$ in Eq.~\eqref{eq:policy_improvement_step} using $\mathcal{D}_{\model} \cup \mathcal{D}_{\mathcal{E}}$
    }
  }
  \caption{Learnable Hamiltonian model-based soft actor critic (LH-MBSAC)}
  \label{algorithm:LH-MBSAC}
\end{algorithm*}

In this paper, we use the soft actor-critic (SAC) algorithm~\cite{haarnoja2018soft} as our base (model-free) RL algorithm. For brevity, we only highlight parts of SAC relevant to us. A detailed description can be found in the original paper~\cite{haarnoja2018soft}. We use a neural network policy function $\pi_\theta(\action_k | \s_k)$, with the optimizable parameters $\theta$, as the actor and the state-action value function $Q_\phi(\s_k,\action_k) = \mathds{E}_{(\s_k, \action_k) \sim \mathcal{E}_\pi}\left[\sum_{k=0}^\infty \gamma^k(\rew(\s_k,\action_k)+\alpha J_1(\s_k))\right]$ as the neural network critic with parameters $\phi$. Both $\pi$ and $Q$ are simple multilayer perceptrons. In essence, the critic is used to reduce the high variance in the reward function due to the non-stationary nature of the MDP. It is trained by having its predictions match the estimated $\hat{Q}$ values obtained for some data $\{\s_k,\s_{k+1},\action_k,\rew_k\}_{k=1}^b$ obtained from a $b$-length rollout (number of interactions) with $\mathcal{E}$. The actor is trained by minimizing the loss function
\begin{align}
  J'(\pi_\theta) = \mathds{E}_{(\s_k,\action_k) \sim \mathcal{E}_{\pi_\theta}}\left[ \alpha \log{\pi_\theta(\action_k|\s_k)} - Q_\phi(\s_k,\action_k)\right],
  \label{eq:policy_improvement_step}
\end{align}
which is equivalent to maximizing $J$ in Eq.~\eqref{eq:rl_objective}. For SAC, this policy optimization is carried out heuristically using neural networks to approximate the policy function $\pi_{\theta}$. We define the number of agent-environment interactions needed to find an approximately optimal policy $\pi^*$ as the \emph{sample complexity}. Moreover, the policy outputs parametrize the mean and covariance $\pmb{\mu}, \pmb{\Sigma}$ of a multivariate Gaussian $\mathcal{N}(\pmb{\mu},\pmb{\Sigma})$ from which the control vector $\mathbf{u}$ is drawn. For the quantum control problem in Eq.~\eqref{eq:discretized_control_problem}, we are usually just concerned with finding an optimal action sequence $\mathbf{u}^*$ producing the maximum intermediate reward $\pmb{r}_k$ rather than the optimal policy function $\pi^*$ which can be produced by a suboptimal policy, too.

SAC can be augmented to incorporate a model $\model(\s_k, \action_k)$ that approximates the dynamics of $\mathcal{E}(\s_k, \action_k)$ using the policy's interaction data $\mathcal{D}$~\cite{MBPO} where $\zeta$ are the model's learnable parameters. The model acts as a proxy for the environment and allows the policy to do MDP rollouts (steps) to augment the interaction data. For this to work, the dynamics obtained from interacting with $\model$ must be close enough to the true dynamics of $\mathcal{E}$ to allow the policy to maximize $J$. By improving the returns $\hat{\eta}(\pi)$ on the model $\model$ by at least a tolerance factor that depends on this dynamical modelling error, the policy's true returns $\eta(\pi)$ on the environment are guaranteed to improve (\cite{MBPO}, see App.~\ref{app:monotonic_returns_improvement} for a detailed mathematical discussion). See Fig.~\ref{fig:algo_schematic} for an illustration of model-based RL. A good choice of the model function class, therefore, can impose strong and beneficial constraints on the space of possible predicted dynamics and thus lead to a smaller modelling error and returns' tolerance factor or allow the model to reduce the tolerance factor greatly after consuming an appropriate amount of training data.

Our choice of the model's functional form is motivated by the two ideas presented in the introduction: (a) incorporating correct partial knowledge about the physical system in the model ansatz parameters; (b) encoding the problem's symmetries and structure into model predictions as function space constraints. For the system in Eq.~\eqref{eq:general_hamiltonian} we assume that the controls are partially characterized to address (a). Specifically, its time-dependent control structure $H_c$ is known. We achieve (b) by parametrizing the system Hamiltonian $H^{(L)}_0(\pmb{\zeta})$ with learnable parameters $\pmb{\zeta}$, where $L$ is the number of qubits. We make the model $\model$ a differentiable ODE whose generator is interpretable and has the form
\begin{align}
  H_{\pmb{\zeta}}(\mathbf{u}(t), t)
  &= H^{(L)}_0(\pmb{\zeta}) + H_c(\mathbf{u}(t), t)\nonumber \\
  &= \sum_{l=1}^{n^2}\zeta_lP_l +  H_c(\mathbf{u}(t), t)
  \label{eq:learnable_system_hamiltonian_pauli_param}
\end{align}
where $\zeta_l = \Tr[P_l H_0(t)] \in [-1,1]$ are real. Generally, like the Choi state, $H_0/\Tr[H_0]$ admits an arbitrary decomposition in terms of a basis $\{P_l\}_{l=1}^{n^2-1}$ of the $\SU{n}$'s Lie algebra. Analogously, for an open system, we parametrize the time-independent part of any dissipation dynamics in addition to the system Hamiltonian using an $\SU(n^2)$ algebra parametrization: $\mathbf{G}_0^{(L)}(\pmb{\zeta}^\text{diss})=\sum_l\zeta_l^\text{diss}P_l$ in the full generator $\mathbf{G}_{\pmb{\zeta}}$.

The model is trained by minimizing the regression loss for single timestep predictions using data uniformly sampled, $D \sim \mathcal{D}$, where $\mathcal{D}$ represents the entire dataset,
\begin{align}
  L_\text{model}(D) = \sum_{D}{\left(\mathbf{M}_{\pmb{\zeta}}\left(\s_k,\action_k\right) - \s_{k+1}\right)^2}.
  \label{eq:regression_loss}
\end{align}
To understand why a differentiable ODE ansatz is a good choice for the model, we need to define an ODE path that is given by $\phi_t: \mathbf{E}(0) \xrightarrow{H_{\pmb{\zeta}}} \mathbf{E}(T)$ generated by $H_{\pmb{\zeta}}$ for some time $t \in [0,T]$ and propagator $\mathbf{E}$. The ansatz is a good choice because of the following two properties of ODE paths: (a) they do not intersect and (b) if paths $\phi_0^{(A)}$, $\phi_0^{(B)}$ start close compared to path $\phi_0^{(C)}$, then paths $\phi_t^{(A)}$, $\phi_t^{(B)}$ remain close compared to path $\phi_t^{(C)}$.

Both properties are well known~\cite{younes_ode_book, howard1998gronwall} for ODEs and become very useful when we try to predict the trajectories from noisy quantum data by imposing strong priors on the space of learnable Hamiltonians. Property (b) is a consequence of Gronwall's inequality~\cite{howard1998gronwall} and essentially can be interpreted as: ODE flows that start off closer (w.r.t. the initial condition) stay closer (w.r.t. the final condition). Both (a) and (b) essentially imply a sort of intrinsic robustness of the ODE flow $\phi_t(\mathbf{z}_0)$ to perturbations on $\mathbf{z}_0$~\cite{intrinsic_robustness_ode}. They constrain the trajectories predicted by the model $\model$ to be intrinsically robust (over a finite time interval) to small noise in the states $\s_k$ and inaccuracies in the learned system Hamiltonian $H^{(L)}_0(\pmb{\zeta})$.

We call the SAC equipped with this differentiable ODE model the learnable Hamiltonian model-based SAC (LH-MBSAC) as listed in Algorithm~\ref{algorithm:LH-MBSAC}. Crucially, LH-MBSAC generalizes the SAC by allowing the policy to interact with the ODE model and the physical system. LH-MBSAC gracefully falls back to the model-free SAC in the absence of a model with low prediction error that is measured from the performance of the model's predictions on an unseen validation set of interaction data. The threshold or tolerance level for switching to the agent-model interaction part of the algorithm is likely problem-dependent and thus needs to be selected along with other hyperparameters in RL. However, this allows us to improve the sample complexity of model-free reinforcement learning, when possible, by leveraging knowledge about the controllable quantum system, yet we are still able to control the system in a model-free manner if this is not possible.

\section{Experiments}\label{sec:results}

We demonstrate the performance of LH-MBSAC on three quantum systems of current interest in open and closed settings with shot noise.
{
Measurements in this section are made using Pauli instead of the generalized Gell-Mann operators mentioned in Sec.~\ref{ssec:open_system_dynamics} and the simulated systems are all qubit systems.}

To warm up, the first system $\H_\text{NV}^{(1)}$ is a single-qubit NV center with microwave pulse control~\cite{nv_center_1q},
\begin{equation}
  \frac{H_\text{NV}^{(1)}(t)}{\hbar}
  = 2\pi \Delta \sigma_z + \underbrace{2\pi \Omega \left( u_1(t)\sigma_x + u_2(t)\sigma_y \right)}_{H_c(t)},
  \label{eq:nv_center_1q}
\end{equation}
where $\Delta=1\text{ MHz}$ is the microwave frequency detuning, $\Omega = 1.4\text{ MHz}$ is the Rabi frequency and the control field parameters are $u_j(t)$ in the range $\mathbb{X}_\text{NV}^{(1)} = \{-1 \leq u_{j} \leq 1\}$. In this and subsequent examples terms not covered by $H_c(t)$ are learned, parametrized by the learnable model parameters $\pmb{\zeta}$. The gate operation time is \SI{20}{\micro\second}.

The second system $H_\text{NV}^{(2)}$ is a two-qubit NV center system~\cite{nv_center}, driven by microwave pulses of approximately \SI{0.5}{MHz}, modelled as follows
\begin{align}
\frac{H_\text{NV}^{(2)}(t)}{\hbar}
  &= \ketbra{1}{1} \otimes \left( -\left(\nu_{z} + a_{zz} \right)\sigma_z - a_{zx} \sigma_x \right) \nonumber\\
  &+ \ketbra{0}{0} \otimes \nu_{z} \sigma_z + \underbrace{\sum_{l=x,y}\sum_{k=1}^2{\sigma^{(l)}_k u_{lk}(t)}}_{H_c(t)},
\label{eq:nitrogen_vacancy_center_ham}
\end{align}
where $\nu_{z} = \SI{0.158}{MHz}$, $a_{zz} = -\SI{0.152}{MHz}$ and $a_{zx} = \SI{-0.11}{ MHz}$, $\sigma^{(l)}_k$ is the $l$th Pauli operator on qubit $k$, and $u_{lk}(t)$ is a time-dependent control field. The range of control is $\mathbb{X}_\text{NV}^{(2)} = \{\SI{-1}{MHz} \leq u_{lk} \leq \SI{1}{MHz}\}$ and the final gate time is $T=\SI{2}{\micro\second}$.

{
The third system $\H_\text{tra}^{(L)}$ is an effective Hamiltonian model for cavity quantum electrodynamics (cQED)~\cite{effective_ham_gambetta} for two transmons or qubits as a proxy for the IBM quantum circuits~\cite{qiskit},
\begin{align}
  \frac{H_\text{tra}^{(2)}(t)}{\hbar}
  &= \sum_{l=1}^2\omega_l \hat{b}_l^\dagger \hat{b}_l + \frac{\eta_l}{2}\hat{b}_l^\dagger \hat{b}_l (\hat{b}_l^\dagger \hat{b}_l - \mathds{1}) \\\nonumber
  &+ J\sum_{l=1}^2(\hat{b}_l^\dagger \hat{b}_{l+1} + \hat{b}_{l} \hat{b}_{l+1}^\dagger) + \underbrace{\sum_{l=1}^2{u_l(t)(\hat{b}_l+\hat{b}_l^\dagger)}}_{H_c(t)}.
  \label{eq:transmon}
\end{align}}

This model consists of Duffing oscillators with frequency $\omega_l=\SI{5}{GHz}$ representing the qubits with an anharmonicity $\eta_l = \SI{0.2}{GHz}$, qubit coupling $J$, and a control field $u_l$ per qubit. This is a special case of the Bose-Hubbard model~\cite{bosehubbard} with $\hat{b}_l$ representing the boson annihilation operator on the $j$th qubit. The control field $u_l(t)$ is real by construction in addition to extra constraints imposed on the space of possible controls $\mathbb{X}$. The range of control is given by $\mathbb{X}_\text{tra}^{(2)}= \{\SI{-0.2}{GHz} \leq u_{l} \leq \SI{0.2}{GHz}\}$ and the final gate time is $T=\SI{20}{\micro\second}$.

\begin{figure*}
\centering
\includegraphics[width=1.\linewidth]{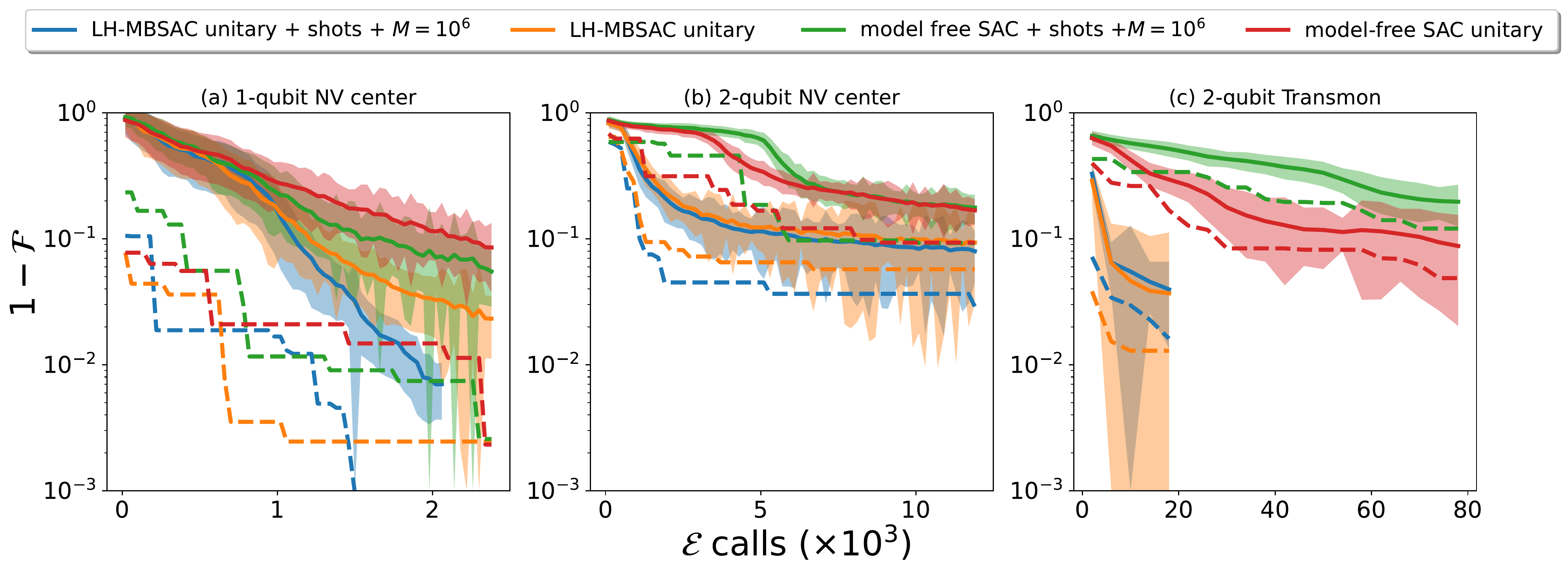}
\caption{The closed system fidelity $\mathcal{F}$ of the Hadamard gate for (a) $H_\text{NV}^{(1)}$, and of the CNOT gate for (b) $H_\text{NV}^{(2)}$ and (c) $H_\text{tra}^{(2)}$ as a function of the number of environment $\mathcal{E}$ calls. The mean fidelity over $100$ controllers is plotted as a solid line with the shading indicating two standard deviations, and the maximum fidelity is indicated by the dashed line. LH-MBSAC or model-free SAC with the unitary tag indicates the shot-noise-free closed system problem in Eq.~\eqref{eq:control_problem_standard} and single shot measurements are indicated likewise.
{We terminate the algorithm early at $\mathcal{F}>0.98$ for LH-MBSAC with and without single shot measurements since the model simulations are expensive and the learned model at this point can be used to further optimize the moderately high fidelity RL pulses further as shown in Sec.~\ref{ssec:graped_sac}.}
The sample complexity of LH-MBSAC is significantly improved for the two-qubit transmon and the NV center over model-free SAC for the closed system control problem and with single shot measurements (of size $M=10^6$), using AAPT. We average these results over three seeds of each algorithm run where a seed refers to a single algorithm run from scratch with a fresh set of randomly initialized parameters.}
\label{fig:sample_complexity_results}
\end{figure*}

For the two-qubit system, the target gate is CNOT and for the one-qubit system, it is the Hadamard gate. Pulses are discretized in accordance with the scheme introduced in Sec.~\ref{ssec:discretization} for the number of timesteps, $N=20$. We follow the parameter restrictions for all systems introduced in Refs.~\cite{c3, effective_ham_gambetta, nv_center, nv_center_1q}. Moreover, due to limited support in our auto-differentiation library~\cite{pytorch}, we simulate the complex dynamics by mapping the complex ODE to two real coupled ODEs~\cite{complex_real_isomorph} (see App.~\ref{app:real_complex_isomorphism} for more details on our ODE solver).

The following sections are organized as follows. In Sec.~\ref{ssec:sample_efficient_closed}, we demonstrate a sample complexity improvement for the different control problems discussed above in a noisy closed setting. For the subsequent sections, we study the two-qubit transmon control problem in more detail. The results were similar for other systems that we studied. In Sec.~\ref{ssec:sample_complexity_vs_delta},
{we study the effect of increasing the estimated Hamiltonian error from its true value on the sample complexity of control.}
Sec.~\ref{ssec:graped_sac} discusses how the learned Hamiltonian in LH-MBSAC can be further utilized for model-based control using gradient-based methods like GRAPE. Sec.~\ref{ssec:open_sys_results} extends results from the closed setting to the noisy open system setting. Finally, in Sec.~\ref{ssec:lims_and_silvers}, we highlight some limitations and silver linings of the LH-MBSAC {and the RL-for-control approach for our specific MDP (Eq.~\eqref{eq:qc_problem_for_rl}) in this paper} and provide promising ideas to circumvent some of the issues.

\subsection{Sample Efficiency for Closed System Control}\label{ssec:sample_efficient_closed}

In this section, we only consider closed or unitary system control with and without single shot measurements defined in Sec.~\ref{ssec:closed-system-dynamics}. From here on, we refer to single shot measurements as just ``shots''.

Unitary control (with closed system dynamics) is implemented for shots as a special case of open system control where the dissipation operator $\mathfrak{L}$ is $0$. The Choi operator $\choi$ corresponding to the gate realized by the controls is obtained by sampling from the binomial distribution in Eq.~\eqref{eq:binomial_gate_measurements} with $M=10^6$ shots per measurement operator. By Hoeffding's inequality~\cite{mohri2018foundations}, we know that with probability $1-0.01$ the error in the estimator of $q_l$ is $10^{-3}$. Or generally, with probability $1-\delta$, for $\epsilon$ error, we require $O(\log{\frac{1}{\delta}}/\epsilon^2)$ measurements. The AAPT method~\cite{ancilla_assisted_quantum_process_tomography} (see Sec.~\ref{ssec:open_system_dynamics}) uses $M \times 3^L$ shots in total for $3^{L}$ possible measurement operators for an $L$-qubit system, which is quite expensive.

Further sparsity restrictions on the structure of $\choi$ imposed by a $k$-local Hamiltonian, where qubit interactions up to only the nearest $k \leq L$ qubits are assumed, can allow the shot cost to go down to $O(4^k(\log{M})/\epsilon^2)$ for $M$ observables due to a reduction in the number of observables that need to be measured or tracked which is
{asymptotically optimal in the number of measurements~\cite{huang2020predicting}. However, since the goal of this paper is gate control, these costs are generally unavoidable to completely verify gate performance. In practice, such gates are only limited to a few qubits and operations on many qubits
are achieved in the circuit formalism through gate composition~\cite{nielsen2010quantum,qiskit}.}

{We randomly initialize the learnable system Hamiltonian using the Pauli basis parametrization in Eq.~\eqref{eq:learnable_system_hamiltonian_pauli_param} with coefficients $\zeta_i \sim \text{Uniform}(-1,1)$. The environment's data buffer $D_{\mathcal{E}}$ that stores the model's training data, i.e., the initial exploration dataset (see Algorithm~\ref{algorithm:LH-MBSAC}), consists of $1$, $20$, and $100$ pulse sequences for the one-qubit NV, two-qubit NV and two-qubit transmon systems respectively. A more detailed discussion of the amount of training data needed for Hamiltonian learning is presented in Appendix~\ref{app:how_much_data}. These data are collected using random uniform policy actions during the first run of the LH-MBSAC algorithm.}

\begin{figure*}
  \centering
  \includegraphics[width=1.\linewidth]{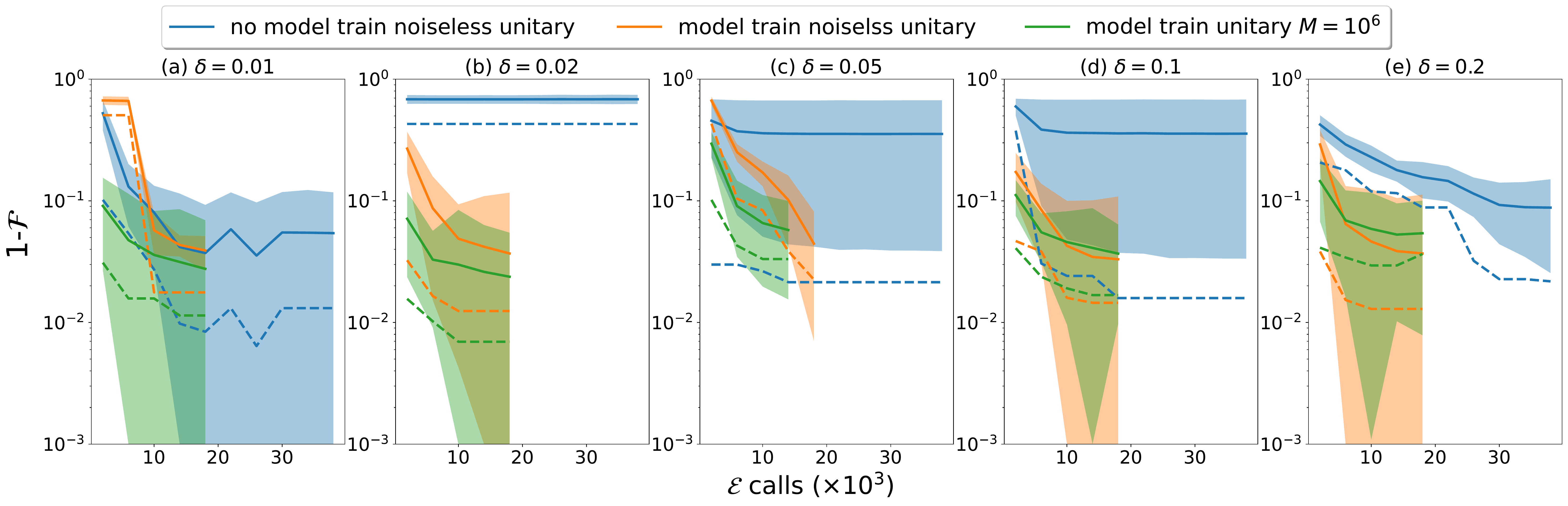}
  \caption{
  {
  Sample complexity or $\mathcal{E}$ calls of LH-MBSAC for the two-qubit transmon control problem as a function of spectral norm error $\delta$, quantifying closeness of the learned system Hamiltonian $H_0(\pmb{\zeta})$ and the true system Hamiltonian $H_0$. The cases for $\delta=0.01, 0.02, 0.05, 0.1, 0.2$ are plotted in (a)--(e). The mean fidelity over $100$ controllers is plotted as a solid line with the shading indicating two standard deviations and the maximum fidelity is indicated by the dashed line. The `noiseless unitary' is the no shot noise setting where the exact unitary is seen by the algorithm while alternatively the unitary is estimated using AAPT with $M=10^6$ shots per observable characterizing the Choi state. The `no model train' line indicates the setting where no learning of $H_0(\pmb{\zeta})$ occurs and $\delta$ is fixed while the `model train' lines denote the setting where $\delta$ is reduced through model training. In general, we see that there are some instances where the RL agent is able to optimize the objectively wrong model $\delta=0.2, 0.01$ and there is a non-linear dependence of $\mathcal{E}$ calls on $\delta$, i.e., a large $\delta$ can produce better model-predictive trajectories with a smaller unitary prediction error. This points us to consider the idea of learning Hamiltonians that are only `locally consistent'. Once learning $H_0(\pmb{\zeta})$ is enabled, algorithmic performance is restored in both the noiseless (with no shot noise) and shot-noise unitary settings. The number of measurements is $M=10^6$ per observable.
  }}
  \label{fig:transmon_hdelta_experiment}
\end{figure*}

{The exploration dataset is then used to learn the system Hamiltonians $H_{0_\text{NV}}^{(1)}$, $H_{0_\text{NV}}^{(2)F}$, $H_{0_\text{tra}}^{(2)}$ via supervised learning of $\model$ using the dynamics prediction loss function (Eq.~\eqref{eq:regression_loss}) until a validation loss of around $10^{-3}\times 2^{2q} \times \texttt{batch\_size}$ is reached, where $\texttt{batch\_size}$ is the number of samples used for a single training policy update. Here $q$ is the number of qubits and $q=2$ for the theoretical unitary and $q=4$ for the Choi state (due to the Choi-Jamiolkowski isomorphism in AAPT). }

{After this, we switch to the model $\model$ to generate synthetic samples to train the policy $\pi$. Whilst concurrently maintaining policy interactions and attempting control of the system via the policy $\pi$, the model is successively trained in periods with fresh data to reduce the model error even further. Once the policy starts producing pulses with nearly optimal fidelities of around $0.98$, we terminate the algorithm and use the learned Hamiltonian to further optimize the pulses using gradient-based methods like GRAPE to (a) reduce sample complexity costs and (b) improve runtime of LH-MBSAC, since the model simulations are computationally expensive. We found that terminating around $0.98$ ensures that the application of further gradient-based methods doesn't cause the control parameters to diverge too much from their initial values thereby retaining, at least partially, their favourable robustness properties \cite{self2}. Step (b) is discussed in detail in Sec.~\ref{ssec:graped_sac}. }

{The results for LH-MBSAC and model-free SAC for the one- and two-qubit control problems are shown in Fig.~\ref{fig:sample_complexity_results}. We consider LH-MBSAC's performance with shots by estimating the gate using its corresponding estimated Choi state $\choi$ using AAPT with $10^6$ shots per observable. The sample complexity of LH-MBSAC to achieve a maximum fidelity significantly improves, by at least an order of magnitude, upon the model-free baseline in both cases, although it is more significant for the two-qubit transmon.
}

\subsection{Sample Complexity as a Function of Hamiltonian Error}
\label{ssec:sample_complexity_vs_delta}

{
Continuing with the closed system control problem, in this section, we study the relationship between sample complexity and error in the estimated model Hamiltonian $H_0(\pmb{\zeta})$ compared to the true system Hamiltonian $H_0$ as the error is increased. This relationship is highly non-linear or irregular and is discussed in detail later in the section. On a high level, the purpose of this section is to understand the interplay between control and model learning especially if the model is inaccurate. Can we still learn a near optimal control policy even if the model is incorrect?
To an extent, yes: we show that when the model error is small, LH-MBSAC is able to successfully find a near optimal control pulse, even with an incorrect model.
% Moreover, the answer to this question is well known in adaptive control systems where it is also yes.

We define the model error $\delta$ as in Ref.~\cite{burgarth2022one}:
\begin{equation}
  \delta = \left\| H_0(\pmb{\zeta}) - H_0 \right\|
  \label{eq:delta_to_true_ham}
\end{equation}
where $\|\cdot\|$ is the spectral norm (the largest singular value) of $H_0(\pmb{\zeta})-H_0$. For this study, we compare two settings for some value of $\delta$ in each experimental run: (i) \emph{learning the system Hamiltonian}, i.e., $\delta$ is decreased from its initial value; (ii) \emph{not learning the system Hamiltonian}, i.e., $\delta$ remains fixed throughout the experiment. Case~(ii) effectively corresponds to Algorithm~\ref{algorithm:LH-MBSAC} without any model training, i.e., we do not attempt to minimize $L_\text{model}(D_\text{train})$ to update the model and instead set the model to have a fixed constant Hamiltonian error $\delta$. The range of Hamiltonians corresponding to different $\delta$ values are chosen by randomly sampling the true Hamiltonian with rejection using Gaussian perturbations. The non-linear dependence on the sample complexity of LH-MBSAC as a function of $\delta$ for the two-qubit transmon control problem for both cases is shown in Fig.~\ref{fig:transmon_hdelta_experiment}(a)--(e) for $\delta= \in \{0.01, 0.02, 0.05, 0.1, 0.2\}$.

For the two-qubit transmon problem, the $\delta=0.02, 0.05, 0.1$ results show worse performance compared to the $\delta=0.2$ results for the theoretical unitary control problem (without measurement noise). This indicates that some model system Hamiltonians $H_0(\pmb{\zeta})$ with a larger $\delta$ predict dynamics more consistent with the true system Hamiltonian $H_0$ dynamics than $H_0(\pmb{\zeta})$ with a smaller $\delta$. However, learning $H_{0_\text{tra}}^{(2)}$ for all shown cases restores performance for both the noiseless unitary and shots-based closed system control problems.

To explain these empirical results and make them more intuitive, we now make use of the integration by parts lemma of Ref.~\cite{burgarth2022one} that bounds $\delta$ by the unitary prediction error of the ODE model w.r.t. the environment for the unitary control problem Eq.~\eqref{eq:control_problem_standard}.

\begin{proposition}\label{prop:model_prediction_unitary_bound_main}
  The following bound holds for the difference between the unitary model's predicted state $U_{\model}$ and the environment's unitary state $U_{\mathcal{E}}$,
  \begin{multline}
    \left\|U_{\mathcal{E}} - U_{\model}\right\|_{\infty,t} \\
      \leq t^2 \delta
      \left(\frac{1}{t}+ \frac{2}{t}\|H_c\|_{1,t}+\|H_{\pmb{\zeta}}\| + \|H_{\mathcal{E}}\|\right)
  \end{multline}
  where $\|\cdot\|$ is the spectral norm and for some linear operator $A$, we have $\|A\|_{\infty, t}=\sup_{s\in[0,t]}\|A(s)\|$ and $\|A\|_{1,t}=\int_0^t ds\|A(s)\|$.
\end{proposition}
\begin{proof}
  See proof of Prop.~\ref{prop:model_prediction_unitary_bound} in App.~\ref{app:bounds}.
\end{proof}

Proposition~\ref{prop:model_prediction_unitary_bound_main} hints at the intuition for why the Hamiltonian error is generally not linearly related to the propagator error.

Although there are some works with better relational bounds on the Hamiltonian error in terms of the observable error, these hinge on the ability to maintain a privileged basis and/or access to special probe states such as the Gibbs state basis~\cite{anshu2021, haah2022}. These bounds crucially do not include the propagator error, thanks to previous assumptions, which is a more general approach to bounding the quantum dynamical evolution error. Of course, there is always a price to be paid for generality and in this case, it is that the error bounds are less constrained and the link between the Hamiltonian and the unitary error becomes non-linear for the general case of the bound.

From Prop.~\ref{prop:model_prediction_unitary_bound_main}, we infer that the unitary model prediction error or the supervised learning regression loss $L_\text{model}(D_{\text{train}})$ in Eq.~\eqref{eq:regression_loss} being small does not imply closeness between learned and true system Hamiltonian, i.e., $\delta \to 0$. However, in the converse case, $\delta$ being very small necessarily implies small propagator error. This is illustrated for the two-qubit transmon Hamiltonian in Fig.~\ref{fig:local_hams_and_graped_conts}(a). The Hamiltonians are again sampled using Gaussian perturbations to the transmon Hamiltonian. There is also significant variation in the unitary model prediction error, even for the same value of $\delta$ for different repetitions of the random Hamiltonian. However, we see that with decreasing $\delta$, the variation decreases, which is also explained by the above bound. Finally, the same pattern can also be observed if we take $\delta$ to be the mean squared difference between the Pauli coefficients of the true and learned Hamiltonian. Thus, this behaviour is general and not limited to the choice of $\delta$.

The main takeaway of this section, that will be taken further in the next section, is that for the control problems considered here it is only necessary to learn models that are `locally consistent' in terms of the unitary trajectories they generate, and small unitary prediction errors can be achieved by models with non-negligibly small $\delta$.
}

\begin{figure*}
  \centering
  \includegraphics[width=2.\columnwidth]{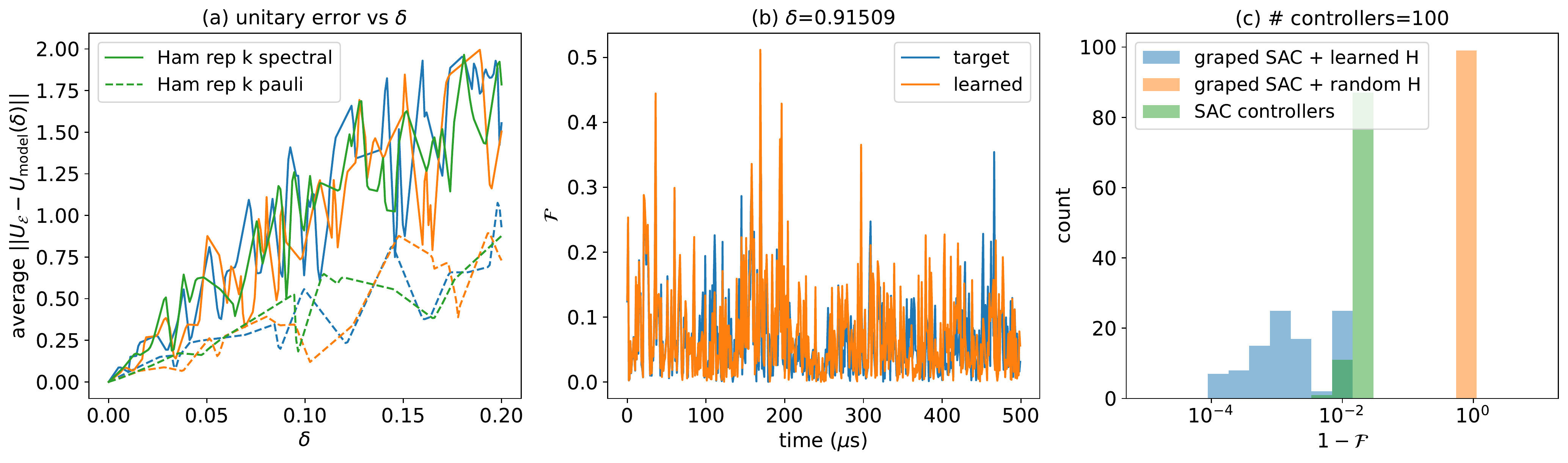}
  \caption{
    {
    (a) An illustration of the non-linear relationship between the unitary model prediction error $\left\|U_{\mathcal{E}} - U_{\model}\right\|$ and Hamiltonian spectral norm (solid) error or mean squared Pauli basis difference (dashed) error as $\delta$ for the two-qubit transmon control problem. For the same $1000$ random control pulses, we evaluate the average unitary prediction error of $\model$ with increasing $\delta$ for three different uniform randomly sampled two-qubit Hamiltonians $H_0(\pmb{\zeta})$ to illustrate the variation in response to the unitary error. (b) Local and global unitary trajectories: $\mathcal{F}$ as a function of a random control pulse with either the learned system Hamiltonian $H_0(\pmb{\zeta})$ or the true system Hamiltonian $H_0$. The learned $H_0(\pmb{\zeta})$ trajectories do not coincide with the global trajectory with $\delta=0.91509$, with the majority contribution coming from a global phase factor such that $\Tr[ H-H_0(\pmb{\zeta})] \approx 0.9$. Both trajectories start off extremely close and start diverging as time increases due to accumulation of small errors in the predicted dynamics. (c) The learned $H_0(\pmb{\zeta})$ can be leveraged using GRAPE to further optimize the fidelities of LH-MBSAC's controllers. We plot a histogram of $100$ LH-MBSAC controller infidelities $1-\mathcal{F}$ before and after applying GRAPE on these controllers using the learned Hamiltonian and a random Hamiltonian. The LH-MBSAC fidelities are significantly improved after applying GRAPE. The appropriate baseline or benchmark representing our ignorance of $H_0$ is a random $H_0(\pmb{\zeta})$ (with uniform random Pauli parameters) which, when plugged into GRAPE, yields extremely low fidelities near $0$ towards the extreme right-hand side of the plot.
    }
  }
  \label{fig:local_hams_and_graped_conts}
\end{figure*}

\subsection{Leveraging the Learned Hamiltonian with GRAPE}\label{ssec:graped_sac}

{
Proposition~\ref{prop:model_prediction_unitary_bound_main} paves the way to learning system Hamiltonians that are locally consistent with the unitary trajectories they generate. By local we mean that the learned Hamiltonian is consistent with the true Hamiltonian on only a subset of all possible generatable trajectories relevant to the control problem. In this section, we delve deeper into the learned model errors and also show that these local models can be leveraged to further optimize the fidelities of LH-MBSAC's controllers using gradient-based methods like GRAPE~\cite{GRAPE,sophie_grape}.

During the model's $\model$ training phase, $H_0(\pmb{\zeta})$ is made consistent with trajectories uniform randomly drawn from the data buffer $D_{\mathcal{E}}$ by minimizing the regression loss $L_\text{model}(D_{\mathcal{E}})$. This allows us to learn a model of the environment that can predict locally consistent unitary trajectories (i.e., at the scale of the control problem). In other words, the learned system Hamiltonian $H_0(\pmb{\zeta})$ does not have to coincide with the true system Hamiltonian $H_0$ for it to be useful for the optimal control task. Indeed, we take the Hamiltonian learned for the two-qubit transmon in Fig.~\ref{fig:sample_complexity_results}(c) and find that it has $\delta=0.91509$. Diving deeper, the matrix difference between the true $H_0$ and learned Hamiltonian $H_0(\pmb{\zeta})$ is,
% we now add the matrix in latex below
\begin{multline*}
    H-H_0(\pmb{\zeta}) =\\ \small \begin{bmatrix}
                    -0.912 & 0.001 & -0.001 & 0.001 \\
                    0.001 & -0.914 & 0.001-0.001i & 0.001+0.001i \\
                    -0.001 & 0.001+0.001i & -0.913 & -0.001 \\
                    0.001 & -0.001-0.001i & -0.001 & -0.914 \\
                \end{bmatrix}.
\end{multline*}
Notably, we can see that most of the error is actually in $\Tr[ H-H_0(\pmb{\zeta})]$ with the true Hamiltonian being learned up to a scale factor of around $0.9$ with the rest of the parameter error being small. This is precisely the global phase error that cannot be learned~\cite{flammia_bayesian_scalabe}.

Despite this discrepancy between the true and learned system Hamiltonians, we find mostly good local agreement between the two random trajectories they induce thanks to the supervised training phase of the model. We show in Fig.~\ref{fig:local_hams_and_graped_conts}(b) the local and global trajectories corresponding to $H_0(\pmb{\zeta})$ and $H_0$ for the two-qubit transmon which shows that the two unitary trajectories w.r.t. the CNOT fidelity are not always coinciding. More specifically, we can see a high overlap in the fidelities induced by random pulses for times between \SI{0}{\micro\second} to around  \SI{100}{\micro\second}. Moreover, the small differences in the generator only start manifesting as the time scales get longer and this can be explained by accruing of small errors in predicted dynamics. This confirms that the unitary model prediction error grows as a function of time. This makes intuitive sense since predictions far into the future, compared to their time-wise preceding counterparts, must necessarily have more built-up error. Furthermore, this learned `local' $H_0(\pmb{\zeta})$ and the controllers found by LH-MBSAC can be used in conjunction with the model-based GRAPE control algorithm~\cite{GRAPE,sophie_grape} to optimize the SAC controller fidelities much more quickly than via just RL alone using accelerated second-order gradient descent. The LH-MBSAC controllers act as seeds, so GRAPE does not move too far away in pulse parameter space compared to where it started. Although not done here, this can also be imposed as an explicit constraint. Note that the question of exactly when to switch over to GRAPE beyond heuristics remains unanswered.   

The fidelities after applying GRAPE are evaluated w.r.t. the true system Hamiltonian $H_0$. Usually LH-MBSAC controllers have moderately high fidelities around $\mathcal{F} > 0.98$ which are improved to $\mathcal{F} > 0.999$. In Fig.~\ref{fig:local_hams_and_graped_conts}(c), we show the RL controllers being optimized further using the learned $H_0(\pmb{\zeta})$ with GRAPE. Experiments in this section for the two-qubit NV center system yield similar results and can be found in App.~\ref{app:graped_sac_nv}.
}

\subsection{Open System Control with Single Shot Measurements}\label{ssec:open_sys_results}

Due to the interpretable nature of our ODE model's ansatz in Eq.~\eqref{eq:learnable_system_hamiltonian_pauli_param}, it is pertinent to ask if two competing but linear terms in the model $\model$ can be learned simultaneously.
{In this section, we find that for our model learning setting, the answer to this question is no. However, this is not general to all problem settings and could potentially be pursued in future work.}

In the previous sections, we only learn one term represented by $H_0(\pmb{\zeta})$. Utilizing the open system formulation of the control problem in Sec.~\ref{ssec:open_system_dynamics}, we consider Lindblad dissipation along with shot noise for the two-qubit transmon control problem in Eq.~\eqref{eq:noisy_gate_control_problem}. Specifically, we consider the decoherence operator $\mathfrak{L}_\text{diss}^{(l)} = \sqrt{\frac{2}{R^*_l}}b_lb_l^\dagger$, acting on the $l$th qubit, and the decay operator $\mathfrak{L}_\text{decay}^{(l)} = \sqrt{\frac{2}{R_l}}b_l$ for $l=1,2$. $R^*_l$ and $R_l$ are the decoherence and decay rates. Both operators are time-independent

{
Alternatively, we can also represent these operators using the adjoint representation but we note that in the context of this learning problem that representation will not make much difference as our algorithm is able to effectively learn the Hamiltonian up to addition of a scalar matrix. However, practically speaking, one can obtain the energy differences of the Hamiltonian via spectroscopy~\cite{ham_spectro_est} which can then be encoded in the eigenvalues of the adjoint representation. It is also possible to learn these eigenvalues using measurements of canonical (Gibbs) states~\cite{anshu2021}.}

We perform experiments for high and low dissipation corresponding to the gate times $R^{{*}^\text{hi}}_l = R^{\text{hi}}_l = 4~\mu\text{s}$, and $R^{{*}^\text{lo}}_l = R^{\text{lo}}_l = \SI{20}{\micro\second}$. Comprising both of these time-independent operators, the Lindblad term $\mathbf{L}_1$ is learned concomitantly with the system Hamiltonian. The results are shown in Fig.~\ref{fig:open_system_results} where the ``learn'' label signifies that $\mathbf{L}_1$ is being learned in addition to the system Hamiltonian $H_0(\pmb{\zeta})$.

\begin{figure}[t]
  \centering
  \includegraphics[width=1\linewidth]{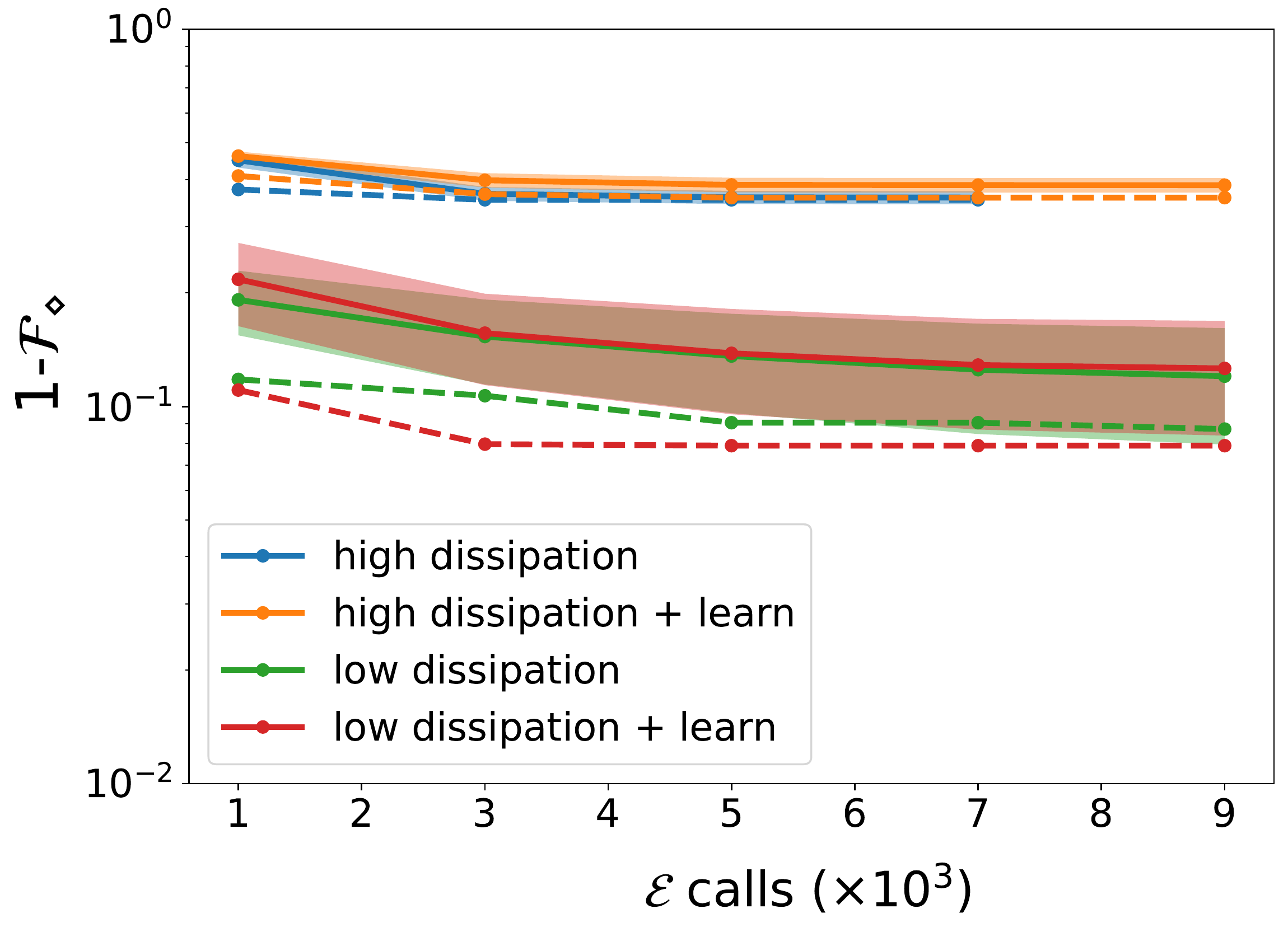}
  \caption{Diamond norm fidelity $\mathcal{F}_\diamond$ for the two-qubit transmon control problem in low and high Lindblad dissipation regimes for LH-MBSAC. The results are averaged over two seeds with the mean $\mathcal{F}_\diamond$ over $100$ controllers shown in solid and the maximum $\mathcal{F}_\diamond$ in dashed lines. Shading denotes two standard deviations from the mean. Here, the `learn' label signifies that dissipation operators are being learned in addition to the system Hamiltonian.}
  \label{fig:open_system_results}
\end{figure}

We use the diamond norm fidelity~\cite{diamondnorm} $\mathcal{F}_\diamond$,
\begin{equation}
  \mathcal{F}_\diamond(\choi(\mathbf{u}(t), t), \choi_\text{target}) = 1-\|\choi(\mathbf{u}(t), t)- \choi_\text{target}\|_\diamond,
  \label{eq:diamondfid}
\end{equation}
instead of the generalised state fidelity since the latter lacks the sensitivity to detect the low dissipation regime (see App.~\ref{app:fid_metrics_comparison}). We find that attempting to learn $\mathbf{L}_1$ while learning $H_0(\pmb{\zeta})$ confers little to no advantage in both the high and low dissipation regimes for this control task. Further investigation shows that the estimate of the system Hamiltonian $H_0(\pmb{\zeta})$ compensates for the observed discrepancy in observed dynamics due to dissipation as much as it is unitarily possible. Moreover, the {learning processes} for $\mathbf{L}_1$ and $H_0(\pmb{\zeta})$ become entangled/mixed so learning multiple independent terms in $\model$ may not be suitable for LH-MBSAC.

\subsection{Limitations and Silver Linings}\label{ssec:lims_and_silvers}

There are two major limitations of LH-MBSAC. The first is that only the system or time-independent part of the Hamiltonian can be learned with the algorithm, while the more difficult problem of learning the time-dependent part of the Hamiltonian~\cite{flammia_bayesian_scalabe} is left as future work.

{Moreover, we found that LH-MBSAC was not able to tackle a three-qubit transmon control problem to obtain a Toffoli gate on an extension of the transmon system. The limitation applied mostly to the RL agent; a viable Hamiltonian is learned that can be leveraged with GRAPE as before. Specific computational details are discussed in App.~\ref{app:three-qubit-transmon}. Essentially, our findings indicate this is an optimization landscape problem and an issue specific to the meta RL strategy of finding optimal pulses instead of a hyperparameter problem. There are two major reasons behind this assessment. Firstly, the values and the gradients for policy and value functions saturate with large training times, i.e., both are stuck in suboptimal extrema, which ultimately culminate with a prematurely optimized reward function. Secondly, since the model Hamiltonian is known beforehand (or also learned), GRAPE equipped with this Hamiltonian and initialized with the highest fidelity LH-MBSAC controllers also gets stuck.}

However, the LH-MBSAC strategy is not limited to SAC and can augment  different RL algorithms for which the three-qubit problem may be tractable.
{Also, since this is likely an optimization landscape issue, a reformulation of the RL control problem could also alleviate this issue by reducing the probability of SAC getting stuck by increasing the range of fidelities the RL agent sees as `proximally optimal'. At present, the agent's goal is to maximize all fidelities it observes, with most of the observations being premature, i.e., before the final gate time. This is highlighted in Fig.~\ref{fig:rl_vs_grapecontrollers} which shows}
the infidelity $1-\mathcal{F}$ as a function of time for $100$ pulses found by LH-MBSAC and GRAPE for the two-qubit transmon control problem. Compared to GRAPE, LH-MBSAC pulses are much more consistent and periodic in terms of the intermediate fidelity values. This highlights that the RL approach is biased towards optimizing intermediate fidelities along with the final target fidelity (since the objective function in Eq.~\eqref{eq:rl_objective} is the {regularized} expected cumulative fidelity). This is quite different from the approach taken by the gradient-based GRAPE algorithm. Despite being interesting from a controller robustness point of view~\cite{self2}, this bias can prevent solutions that do not admit high intermediate fidelities from being found as RL can get stuck in a loop mining medium-level fidelity values. Stepping away from this particular sequential decision-making MDP formulation might be one solution to consider in future work.

There are silver linings for the aforementioned MDP formulation. RL pulses are fidelity-wise better, on average, across the duration of the pulse. Leveraging the learned system Hamiltonian, we can further improve the performance of the RL pulses by using GRAPE with the RL pulse parameters as initialization. As seen in Fig.~\ref{fig:rl_vs_grapecontrollers}, these pulses are still better than the ones found by GRAPE using the learned system Hamiltonian but with completely random pulse initializations, i.e., without LH-MBSAC controllers as seeds.

\begin{figure}[t]
  \centering
  \includegraphics[width=0.95\linewidth]{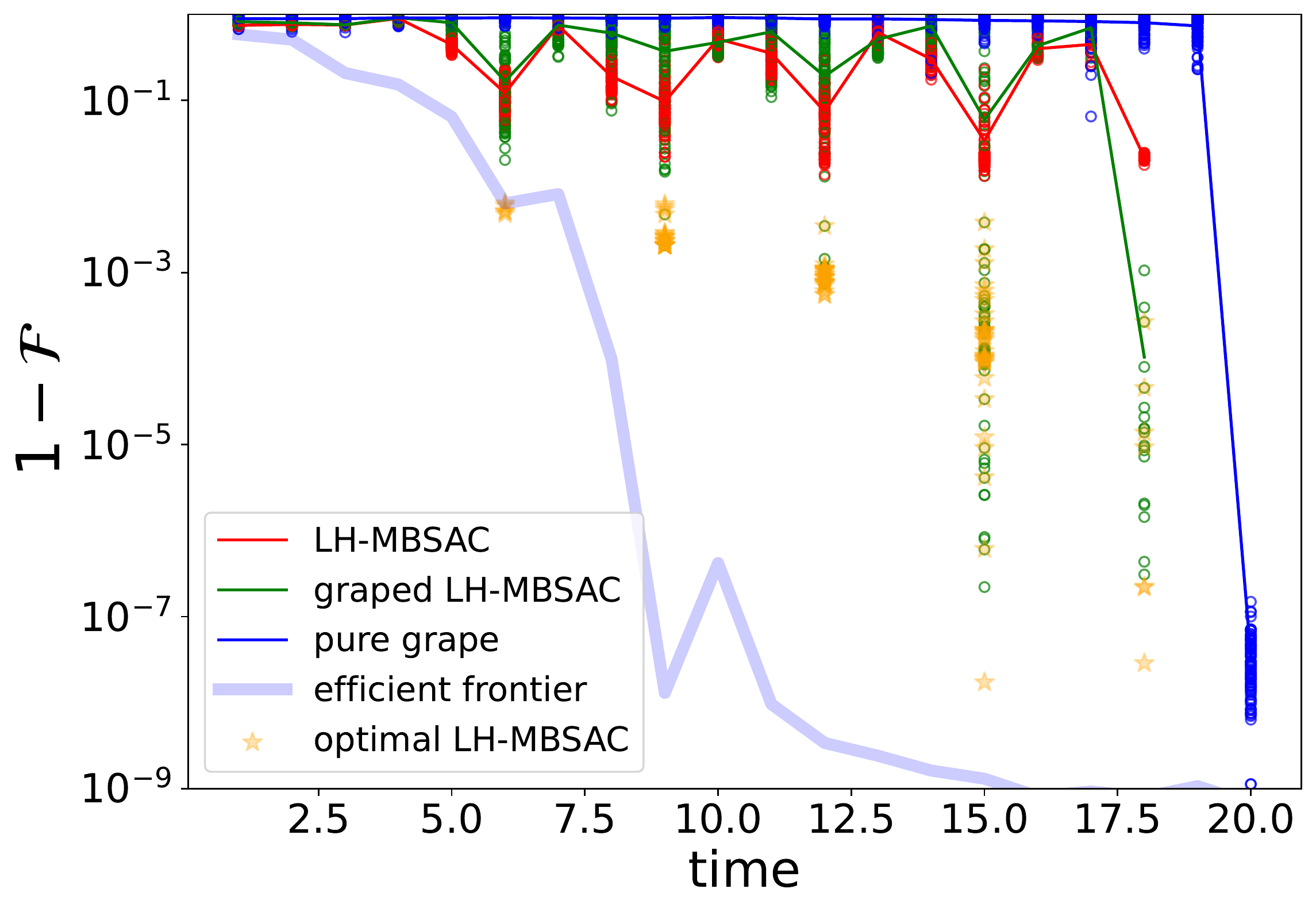}
  \caption{The infidelities over time for $100$ different control pulses found by LH-MBSAC and by GRAPE using the learned system Hamiltonian $H_0(\pmb{\zeta})$ for the two-qubit transmon control problem with final time $T \leq 20~\mu\text{s}$. RL pulses are further optimized using GRAPE. GRAPE is also used to obtain pulses without the RL controls as initial values for a fixed final gate time $T = 20~\mu\text{s}$. Short optimal controls found by RL are identified by truncating RL pulse parameters at times $t \geq \left\{6,9\right\}~\mu\text{s}$ whose final infidelities are shown as stars with $t=6~\mu\text{s}$ being Pareto optimal w.r.t. the efficient frontier (the surface indicating the best fidelity for that time).}
  \label{fig:rl_vs_grapecontrollers}
\end{figure}

Furthermore, this RL bias towards valuing intermediate fidelities allows us to identify optimal pulses that can be executed in short times, which is a difficult problem for GRAPE even if the final gate time is explicitly added to the control objective~\cite{sophie_grape}.

Truncating the control sequence for pulses at time $t$ if the infidelity is below $5 \times 10^{-2}$, we again leverage GRAPE to maximize the final fidelities at these shorter times. These are shown as stars in Fig.~\ref{fig:rl_vs_grapecontrollers} with the fidelities at $t=6~\mu\text{s}$ being approximately Pareto optimal, i.e., the best fidelity for that time. The Pareto optimal efficient frontier is constructed by sampling $100$ GRAPE pulses with random intializations at different final gate times.

% An alternative measure is the diamond norm~\cite{diamondnorm, Lindbladtomo} $\| \Phi(E_\mathbf{I}) - \Phi(V)\|_\diamond$ which is bit more involved to compute numerically since it involves solving a convex optimization problem. All these measures are invariant to arbitrary unitary operations due to the cyclic property of the trace.

\section{Conclusion}

We have presented a learnable Hamiltonian soft actor-critic (LH-MBSAC) algorithm for time-dependent noisy quantum gate control. LH-MBSAC augments model-free soft-actor critic by allowing the reinforcement learning (RL) policy to query a learnable model of the environment or the controllable system. It thereby reduces the total number of queries (sample complexity) required to solve the RL task. The model is a differentiable ODE with a partially characterized Hamiltonian, where only the parametrized time-independent system Hamiltonian is required to be learned. This is a good inductive bias for the quantum control task as ODE trajectories do not intersect, and the Schr\"odinger ordinary differential equation (ODE) preserves unitary evolution, thereby sensibly constraining the space of models to be learned. Using exploration data acquired from the policy during the RL loop, we train the model by reducing a model prediction error over the data. We show that LH-MBSAC is able to reduce the sample complexity for gate control of one- and two-qubit nitrogen-vacancy (NV) centers and transmon systems in unitary and single-shot measurement settings.

Moreover, we highlight that despite the {generally} non-linear relationship between the error in the learned Hamiltonian and the model prediction error, LH-MBSAC's performance is robust to this variation. Furthermore, even if the learned Hamiltonian that minimizes the model prediction error is not the same as the true system Hamiltonian, the learned Hamiltonian {which is locally consistent in terms of its dynamical predictions} can be leveraged using gradient-based methods that require full knowledge of the controllable system, like GRAPE, to further optimize the controllers found by LH-MBSAC. Applying LH-MBSAC in high and low Lindblad dissipation regimes with shot noise, we found that its performance in both was not improved if the Lindblad dissipation terms are also learned in addition to the system Hamiltonian as it is likely that the latter part compensates for the extra dissipation effects.

Despite LH-MBSAC's limitations requiring it to know the time-dependent Hamiltonian and system scalability beyond two qubits (four with single shot measurements due to ancilla assisted process tomography (AAPT)), the algorithm can be used to augment many existing model-free RL approaches for quantum control. This should afford more sample-efficient RL-based optimization of quantum dynamics for near-term noisy quantum processors on a variety of architectures as shown in the paper. Specific tasks can include noisy small circuit optimization, state preparation~\cite{phys_rev_single_shot_reward,bukov} or gate optimization using a partially known model of the underlying dynamics~\cite{dalgaard2020global}. Since having an accurate model can be extremely useful for validation of quantum operations and model bias can be crippling, model-based RL methods like LH-MBSAC can improve the model specifically tailored for some downstream task, e.g., quality assessment of topological codes~\cite{valenti2019} or fine-tuning current implementations of a two-qubit cross resonance gate on some novel architecture~\cite{ding2023} using a pre-existing but partially correct model. Here, the goal for the RL agent would be to help learn effective and potentially scalable models of the target system whilst optimizing the target functional. Another interesting goal in this direction could just be incorporating the number of measurements or queries of the system in the RL objective so that the learning is sample-efficient. Another avenue of future work is to combine LH-MBSAC with a more feasible measurement protocol than AAPT. AAPT is not a hard requirement for our approach and was used here for its theoretically simple estimation of a quantum process. Two angles of attack are either sparsity assumptions on the dynamics generator~\cite{huang2022learninghams} and the generated evolution~\cite{huang2020predicting} or a partially observed Markov Decision Process  formulation of the control problem~\cite{hausknecht2015deep,self1}.

Moreover, despite the scalability problems due to the potentially hindering nature of the RL strategy towards maximizing intermediate fidelities, it can be useful in particular to identify short time optimal pulses. Learning the time-dependent part of the Hamiltonian is harder and might require a stronger learning protocol, e.g., using the zero-order hold method with the learning protocol presented in this paper, Bayesian Hamiltonian Learning~\cite{flammia_bayesian_scalabe} or more informative learning process or Hamiltonian learning methods~\cite{huang2022learninghams,huang2022learningprocesses} which would be exciting to pursue in the future.

The study of the abilities and limitations of our Hamiltonian learning protocol using ZOH will be left to future work. Our code is available at~\cite{code}.

\bibliography{references}

\cleardoublepage
\appendix

\begin{center}
	{\Large \bf Appendices}
\end{center}

\bigskip

Here we present additional details and proofs for the results in the main text.

\section{Mapping Complex Linear ODEs to Coupled Real ODEs and Step-size Effects}\label{app:real_complex_isomorphism}

\begin{figure*}[t]
  \centering
  \includegraphics[width=1.\linewidth]{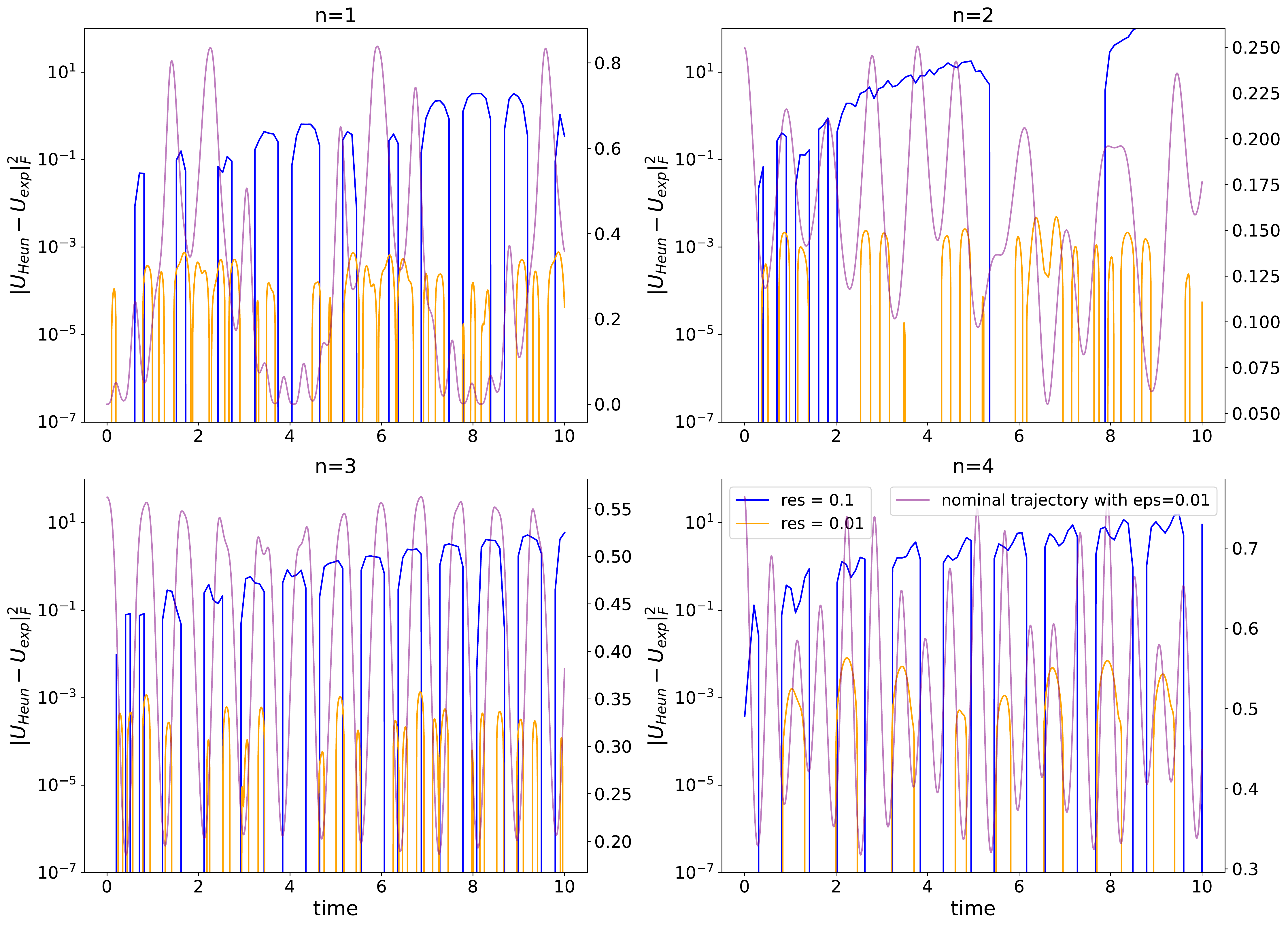}
  \caption{Frobenius norm of the prediction error of the Heun ODE solver~\cite{numerical_analysis_heun} compared to the matrix exponential method. The number of qubits $n$ are shown on top of each subfigure. The random time-dependent sinusoidal Hamiltonians  are as follows:
  for $n=1$, $H=-2.32\sigma_z \cos{2.19t} - 0.01\mathds{1}\sin{3.62t} + 1.79\sigma_x \cos{4.89t} + 3.04\sigma_y\cos{2.69t}$;
  for $n=2$, $H=1.01\sigma_z \mathds{1}\cos{1.44t}+4.5\mathds{1}\mathds{1}\sin{4.55t}-2.7\sigma_y\sigma_z\sin{1.07t}+0.48\sigma_x\sigma_z\cos{2.26t}$;
  for $n=3$, $H=-1.28\mathds{1}\sigma_x\mathds{1}\cos{2.62t}-0.23\sigma_y\sigma_z\sigma_y\sin{3.75t}-1.34\mathds{1}\sigma_y\sigma_x\sin{3.35t}+3.38\sigma_x\sigma_x\sigma_z\cos{2.34t}$;
  for $n=4$, $H=-0.41\mathds{1}\sigma_z\sigma_z\sigma_x\sin{2.86t}+2.19\sigma_y\mathds{1}\sigma_x \sigma_z \sin{1.38t}-0.87\sigma_y\sigma_x\sigma_x\sigma_z\sin{2.26t}+4.06\sigma_x\sigma_x\sigma_z\mathds{1}\sin{1.76t}$ where the shorthand used is $\mathds{1}\sigma_x\mathds{1} \equiv \mathds{1} \otimes \sigma_x \otimes \mathds{1}$.
  e.g. Trace fidelities w.r.t. the generalized CNOT (NOT or X-gate for $n=1$, CNOT for $n=2$, CCNOT for $n=3$ and so on) are shown in the twin axis on the right. It can be seen that the step size of $10^{-1}$ leads to quick accumulation of error seen in the sharp peaks but a step size of $10^{-2}$ is more stable with more than $O(10^3)$ times less prediction error.}
  \label{fig:heun_solver_tests}
\end{figure*}

The quantum control problem in Eqs.~\eqref{eq:control_problem_standard} and Eq.~\eqref{eq:noisy_gate_control_problem} involve ODEs (Eqs.~\eqref{eq:unitary_ode}, \eqref{eq:superoperator_ode}) in the complex domain with a complex vector field map $f_\theta: \mathbb{R} \times \mathbb{C}^d \rightarrow \mathbb{C}^d$ (where $\theta$ denotes some learnable parameters that can be optimized). For the unitary control problem we have a linear map $f_\theta(U(\mathbf{u}(t), t), t) = H_\theta(\mathbf{u}(t), t)U(\mathbf{u}(t), t)$ where $H_\theta$ is a Hermitian Hamiltonian that generates the ODE path of the propagator $U(t)$. We make use of the following isomorphism to map the complex ODE to two coupled real ODEs in $\mathbb{R}^{2d}$ by separating the propagator into its real and imaginary parts $U = U_\text{real} + iU_\text{imag}$ and mapping the Hamiltonian isomorphically $H(\mathbf{u}(t), t) \xrightarrow{\sim}\mathds{1} \otimes H_\text{real}(\mathbf{u}(t), t) - i\sigma_y \otimes H_\text{imag}(\mathbf{u}(t), t)$, to get the following~\cite{complex_real_isomorph} coupled real ODE system,
\begin{equation}
\begin{multlined}
  \dv{}{t}
  \begin{pmatrix}
    U_\text{real}(\mathbf{u}(t), t) \\
    U_\text{imag}(\mathbf{u}(t), t) \\
  \end{pmatrix}\\
  =
  \begin{pmatrix}
    H_\text{imag}(\mathbf{u}(t), t) & H_\text{real}(\mathbf{u}(t), t) \\
    -H_\text{real}(\mathbf{u}(t), t) & H_\text{imag}(\mathbf{u}(t), t) \\
  \end{pmatrix}
  \begin{pmatrix}
    U_\text{real}(\mathbf{u}(t), t) \\
    U_\text{imag}(\mathbf{u}(t), t) \\
  \end{pmatrix}.
  \end{multlined}
  \label{eq:complex_to_coupled_real_ode}
\end{equation}
The mapping is analogous for the superoperator ODE in Eq.~\eqref{eq:superoperator_ode}. Likewise, various other metrics, e.g., fidelity $\mathcal{F}$, were analogously transformed. We made use of the real nature of the Pauli vector decomposition of $H$ to keep track of both the time-independent learnable Hamiltonian and the time-dependent control Hamiltonian representations.

We use Heun's method~\cite{numerical_analysis_heun} to implement a custom differentiable numerical ODE solver in \texttt{pytorch}~\cite{pytorch}, a popular automatic differentiation code library. The solver is able to evolve multiple ODEs under multiple generators in parallel using generalized matrix/tensor operations (ideally on a GPU to maximally leverage computational efficiency). The solver can be accessed in the \texttt{LearnableHamiltonian} module in our code~\cite{code}. To determine the optimal tradeoff between accuracy of dynamical simulation, computed gradients and the size of the computation graph that is held in memory for automatic differentiation, we conduct experiments by simulating the dynamics of random $n$-qubit Hamiltonians from $n=1$ to $n=4$ at different precision or tolerance or step size of the ODE solver (see Fig.~\ref{fig:heun_solver_tests}).

Computational speed of the solver naturally trades off with the accuracy in the simulation and the computed gradients. We find that a step size of $10^{-2}$ is sufficiently accurate for forward dynamical simulation (no gradients are computed in this step) and a step size of $5 \times 10^{-4}$ is required for the backward step when the gradients need to be computed to train the ODE model. The errors in the dynamical predictions (averaged over many thousands of data points) in both steps are reasonably small and monitored. The ODE solvers in \texttt{scipy}~\cite{scipy} and the matrix exponential method for solving linear ODEs~\cite{sophie_grape} both have similar errors than our method for the step size $5 \times 10^{-4}$ (likely the Bayes' optimal error for our numerical simulation).

The ability to be fast, but produce slightly less accurate predictions improved the wall time of our algorithm. Specifically, a significantly large number of trajectories can be quickly sampled in the forward step to augment the RL policy's training data while the much slower backward step can be limited to a smaller number of trajectories that need to be predicted and are divided over multiple batches.

\section{Bounds on the Model Prediction Error}\label{app:bounds}

Consider a unitary RL control problem with the MDP in Eq.~\eqref{eq:qc_problem_for_rl}, where the environment's Hamiltonian and propagator at some timestep $t_l$ are given by $H_{\mathcal{E}}(t_l, u_l) = H_0 + H_c(u_l, t_l)$ and  $U_{\mathcal{E}}(\mathbf{u}_{k})$. Now consider the model $\model(\s_{k+1}|\action_k, \s_k)$ that predicts a single step of unitary dynamics $\s_k \xrightarrow{H_{\pmb{\zeta}}} \s_{k+1}$ under its parametrized generator $H_{\pmb{\zeta}} = H^{(L)}_0(\pmb{\zeta}) + H_c(u_l, t_l)$ following our assumptions in Sec.~\ref{sec:mbrl_control_setup}. Now we bound the error in the single step predicted propagator $U_{\pmb{\zeta}}$ using the integration-by-parts lemma from Ref.~\cite{burgarth2022one}. We consider a continuous version of the propagators and the generators since the result is only used qualitatively.
\begin{proposition}\label{prop:model_prediction_unitary_bound}
  (Bound on the model predictions) The following bound between the unitary model's predicted state $U_{\pmb{\zeta}}(\mathbf{u}_{:k})$ and the environment's unitary state $U_{\mathcal{E}}(\mathbf{u}_k)$ holds,
  \begin{multline}
    \left\|U_{\mathcal{E}} - U_{\model}\right\|_{\infty,t}
    \leq t^2 \left\| H^{(L)}_0(\pmb{\zeta}) - H_0 \right\| \\
      \cdot\left(\frac{1}{t}+\frac{2}{t}\|H_c\|_{1,t}+\|H_{\pmb{\zeta}}\| + \|H_{\mathcal{E}}\|\right).
      \label{eq:unitary_model_prediction_error}
  \end{multline}
\end{proposition}
\begin{specialproof}
The generator difference $H_{\pmb{\zeta}} - H_{\mathcal{E}} = H^{(L)}_0(\pmb{\zeta}) - H_0$ is time-independent. So the integral action difference term becomes
\begin{align}\nonumber
  \left\|\int_0^tds~H^{(L)}_0(\pmb{\zeta}) - H_0 \right\|_{\infty, t}
  &= t \left\|H^{(L)}_0(\pmb{\zeta}) - H_0 \right\|_{\infty, t}\\
  &= t \| H^{(L)}_0(\pmb{\zeta}) - H_0 \|,
\end{align}
where in the last line, we drop the supremum over time due to time independence. Now we can rewrite
\begin{align}
  \|H_{\mathcal{E}}(\mathbf{u}(t), t)\|_{1,t}
  &= t\|H_0 + H_c(\mathbf{u}(t), t)\|_{1,t}  \nonumber\\
  &\leq t\left(\|H_0\| + \|H_c(\mathbf{u}(t), t)\|_{1,t} \right)
\end{align}
using the triangle inequality. Combining both facts yields the inequality.
\end{specialproof}

The inequality in Eq.~\eqref{eq:unitary_model_prediction_error} can be analogously extended to the open system setting w.r.t. the Choi matrix $\choi$. Here, we focus on the unitary case for simplicity since the arguments are similar.

There are two observations worth mentioning about inequality Eq.~\eqref{eq:unitary_model_prediction_error}: (a) when all other variables are fixed, the error in the model's unitary predictions w.r.t. to the environment's ground truth grows as a function of time; (b) the model prediction error is a lower bound of the error in the model parameters $H_0({\pmb{\zeta}})^{(L)}$ w.r.t. the ground truth parameters $H_0$. The prediction error $L_\text{model}(\mathcal{D}_\text{val})$ can be estimated using a validation dataset $\mathcal{D}_\text{val}$ and relates this observed validation loss to the Hamiltonian difference. Importantly, the inequality implies that the closeness in the propagator does not always translate to closeness in the Hamiltonian. Therefore, a model Hamiltonian can be locally a good fit for propagator predictions while still having a large Hamiltonian error $\left\| H^{(L)}_0(\pmb{\zeta}) - H_0 \right\|$. So arbitrary closeness in terms of the Hamiltonian error need not be necessary for good unitary predictions. But conversely, if we can be certain that the model Hamiltonian is close to the system Hamiltonian, then the unitaries must be close. This motivates that a good guess (in the form of partial knowledge about the system) of the true Hamiltonian is useful in bounding the prediction errors.

We exploit this fact to learn the local Hamiltonian $H^{(L)}_0(\pmb{\zeta})$ that approximates the dynamics of $H_0$ w.r.t.\ $U_{\mathcal{E}}$. Qualitatively, we observe that Hamiltonian error, propagator validation and training error are both improved during training (i.e., the propagator loss on the validation set is predictive of Hamiltonian error). This can be seen in Fig.~\ref{fig:datascaling_exp} for the noisy shot setting. But we also note in this example that the learned Hamiltonian $H^{(L)}_0(\pmb{\zeta})$ is local, as seen from the Hamiltonian error plateauing at a non-zero value.

\begin{figure*}[t]
  \centering
  \includegraphics[width=1\linewidth]{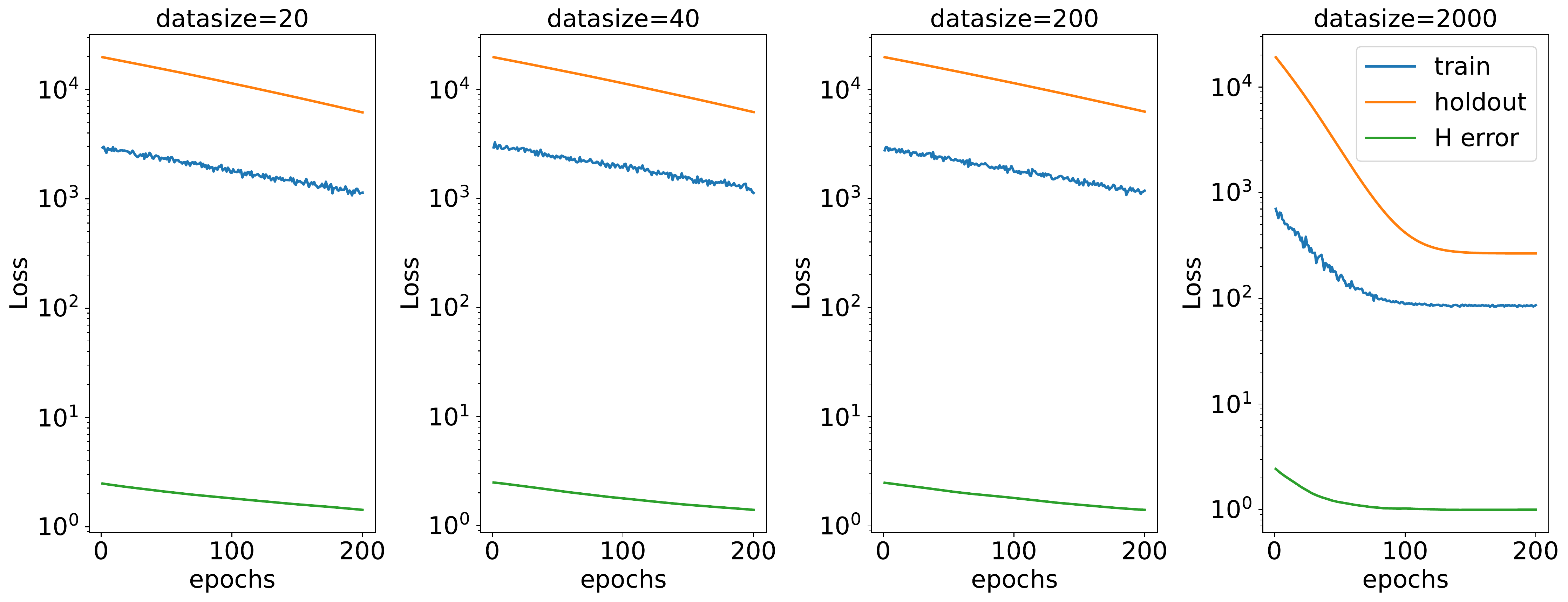}
  \caption{The Hamiltonian error, unitary training $L_\text{model}(\mathcal{D}_\text{train})$ and validation (holdout) loss $L_\text{model}(\mathcal{D}_\text{val})$ as functions of training epochs for the two-qubit transmon unitary control problem with noisy measurements and $M=10^5$. Data size denotes the number of single-step unitary transitions. The validation set is fixed to $5000$ transitions under random policy actions $\action_k$. All three error measures improve as a function of training. Adding more training data appears to provide diminishing returns in predicting the local unitary dynamics.}
  \label{fig:datascaling_exp}
\end{figure*}

\section{Monotonic Improvement for Model Returns}\label{app:monotonic_returns_improvement}

We show that it is possible to improve the environment's reward under an incorrect model ansatz in $\model$. For that we need the following result from~\cite{MBPO},
\begin{theorem}\label{thm:mbpo_performance_bound}
  (Monotonic improvement for model-based returns~\cite{MBPO})
  Given $k$-branch rollout returns $\eta_\text{branch}(\pi)$ for a policy $\pi$ under the model, the true returns $\eta(\pi)$ are lower bounded
  \begin{multline}
    \eta(\pi) \geq \eta_\text{branch}(\pi)\\
     - 2r_\text{max}\left(\frac{\gamma^{k+1}\epsilon_\pi}{(1-\gamma)^2}+\frac{\gamma^k+2}{1-\gamma}\epsilon_\pi + \frac{k}{1-\gamma}\left(\epsilon_\text{model}\right)\right)
    \label{eq:dynamic_model_returns_bound}
  \end{multline}
  where the returns $\eta$ are defined as
  \begin{align}\nonumber
    \eta(\pi) &:= \mathds{E}_{\pi}\left[\sum_{t=0}^{\infty} \gamma^t \rew_t(\s_t,\action_t) \right]\\
              & = \mathds{E}_{\rew_t \sim \mathcal{E}(\s_{t-1},\action_{t}^{\pi})}\left[\sum_{t=0}^{\infty} \gamma^t \rew_t(\s_t,\action_t) \right].
  \label{eq:returns_def}
  \end{align}
  $\rew_\text{max}$ is the maximum reward for an MDP transition; the policy error $\epsilon_\pi$ is the upper bound,
  \begin{align}
    \epsilon_{\pi} \geq D_\text{TV}\left( \pi_D(\s,\action) \| \pi(\s,\action) \right)
    \label{eq:policy_error}
  \end{align}
  where $D_\text{TV}$ is the total variation distance and $\pi_D$ is the data generating policy (i.e., the policy that generated the MDP data by interacting with the environment $\mathcal{E}$). The model error $\epsilon_\text{model}$ is the upper bound
  \begin{align}
    \epsilon_\text{model} \geq
    \max_{t} \left( \mathds{E}_{\s \sim \pi_D^{(t)}}\left[ D_\text{TV}\left( P_\mathcal{E}(\s'|\s,\action) \| P_M(\s'|\s,\action) \right)\right]\right),
  \label{eq:model_error}
  \end{align}
  where $P_M(\s'|\s,\action)$ is the MDP transition probability distribution under the model $M$ that estimates the environment $\mathcal{E}$ and likewise for $P_\mathcal{E}$. $\gamma$ is the discount factor and $k$ is the branch rollout length.
\end{theorem}
\begin{specialproof}
  See proof of Theorem~4.3 in~\cite{MBPO}.
\end{specialproof}

Informally, the theorem states that as long as the returns under the model $\eta_\text{branch}$ are improved by at least the tolerance term $2r_\text{max}(\cdots)$, then the returns under the environment $\eta$ are guaranteed to improve. This also assumes that the policy $\pi$ generating the model returns is reasonably close to the policy that interacts with the environment to generate the MDP data that we use to compute the statistics including the returns. This policy error $\epsilon_\pi$ can be monitored online and controlled while running the algorithm by curtailing its training once it exceeds some tolerance threshold. Moreover, Ref.~\cite{MBPO} shows that as long as the dataset size is large enough, the model error $\epsilon_m$ can be decoupled from the policy error $\epsilon_\pi$. The optimal branch rollout length $k^*$ is given by the minimizer of the tolerance. In practice, there are other considerations (e.g., the interplay between various hyperparameters) that need to be accounted for to determine $k^*$, so it is usually tuned numerically.

Using Thm.~\ref{thm:mbpo_performance_bound} for the ODE model, we can indirectly connect the Hamiltonian error using the validation loss $L_\text{model}(\mathcal{D}_\text{val})$) with $\epsilon_\text{model}$. If the Hamiltonian error is small, then $\epsilon_\text{model}$ is small and the returns from the model and the environment are similar for any interacting policy $\pi_{\nntheta}$. However, the returns need not be exactly the same and just need to be better than the tolerance provided by the term $-2r_\text{max}(\cdots)$ in Eq.~\eqref{eq:dynamic_model_returns_bound} which is a function of $\epsilon_\text{model}$. The tolerance is smaller for a more accurate model and so less of an improvement of the model returns $\eta_\text{branch}$ is necessary. The following lemma makes this idea concrete by applying Thm.~\ref{thm:mbpo_performance_bound} to our RL control problem setup.

\begin{lemma}\label{lemma:model_error_ode}
(Model error upper bound for the ODE model)
If the model error $\epsilon_\text{model}$ upper bounds the risk,
\begin{equation}
    \epsilon_\text{model} \geq \max_t\left(\mathds{E}_{\s \sim \pi_D^{(t)}}\left[ \mathds{I}(\model(\s, \action) \neq \mathcal{E}(\s, \action))\right]\right)
\end{equation}
then it also upper bounds the unitary prediction error
\begin{align}
  \epsilon_\text{model} &\geq
  \max_{t} \left( \mathds{E}_{\s \sim \pi_D^{(t)}}\left[\left\|U_{\mathcal{E}(\s,\action)} - U_{\model(\s,\action)}\right\|_{\infty,t}\right]\right)
    \label{eq:model_error_in_terms_of_unitary_loss}
\end{align}
and the total variation distance between the model and environment probabilistic distributions,
\begin{multline}
  \epsilon_\text{model} \\\geq
    \max_{t} \left( \mathds{E}_{\s \sim \pi_D^{(t)}}\left[ D_\text{TV}\left( P_\mathcal{E}(\s'|\s,\action) \| P_{\model}(\s'|\s,\action) \right)\right]\right).
\end{multline}
% Need to link the binomial random variable with an entry in the unitary. Then the max diff in probabilities under the true and the learned model is still determined by the burgath inequality.
\end{lemma}
\begin{specialproof}
  Since the model $\model$ and the environment are both deterministic by assumption, we need to modify the lower bound on the model error $\epsilon_\text{model}$ in Thm.~\ref{thm:mbpo_performance_bound}. We can replace the total variation distance between the two supposed distributions $P_{\mathcal{E}}, P_{\model}$ by an indicator variable $\mathds{I}(\model(\s, \action) \neq \mathcal{E}(\s, \action))$ if $\s'_{\model} \neq \s'_{\mathcal{E}}$, which is $1$ if the transitioned states do not match and $0$ if they do. We can upper bound the total variation distance like this since $D_\text{TV}(P_{\mathcal{E}}, P_{\model}) = \sup_{A}{| P_{\mathcal{E}}(A) - P_{\model}(A)|} \leq 1$ in case the probabilities do not match and $D_\text{TV}(P_{\mathcal{E}}, P_{\model}) = 0$ when they match perfectly. Hence, there exists some $\epsilon_\text{model}$ such that
  \begin{multline}
    \epsilon_\text{model} \geq
    \max_{t} \left( \mathds{E}_{\s \sim \pi_D^{(t)}}\left[ \mathds{I}(\model(\s, \action) \neq \mathcal{E}(\s, \action))\right]\right) \\\nonumber
    \geq \max_{t} \left( \mathds{E}_{\s \sim \pi_D^{(t)}}\left[ D_\text{TV}\left( P_\mathcal{E}(\s'|\s,\action) \| P_M(\s'|\s,\action) \right)\right]\right).
  \end{multline}
  The risk $ \mathds{E}_{\s \sim \pi_D^{(t)}}\left[ \mathds{I}(\model(\s, \action) \neq \mathcal{E}(\s, \action))\right]$ is essentially the fraction of unitaries that the model predicts incorrectly and is related to the unitary error in Prop.~\ref{prop:model_prediction_unitary_bound} by the fact that
  \begin{align}
    \left\|U_{\mathcal{E}} - U_{\model}\right\|_{\infty,t} \leq \mathds{I}(\model(\s, \action) \neq \mathcal{E}(\s, \action)),
  \end{align}
  provided that $\left\|U_{\mathcal{E}} - U_{\model}\right\|_{\infty,t}$ is normalised to be in $[0,1]$. So we have
  \begin{multline}
    \mathds{E}_{\s \sim \pi_D^{(t)}}\left[\left\|U_{\mathcal{E}(\s,\action)} - U_{\model(\s,\action)}\right\|_{\infty,t}\right] \\
    \leq \mathds{E}_{\s \sim \pi_D^{(t)}}\left[\mathds{I}(\model(\s, \action) \neq \mathcal{E}(\s, \action))\right].
  \end{multline}
  So $\epsilon_\text{model}$ upper bounds the expected unitary error if and only if $\epsilon_\text{model}$ upper bounds the expected risk in the unitary prediction error.
\end{specialproof}

\section{How Much Data is Needed for Model Training?}\label{app:how_much_data}

\begin{figure}
  \centering
  \includegraphics[width=1\linewidth]{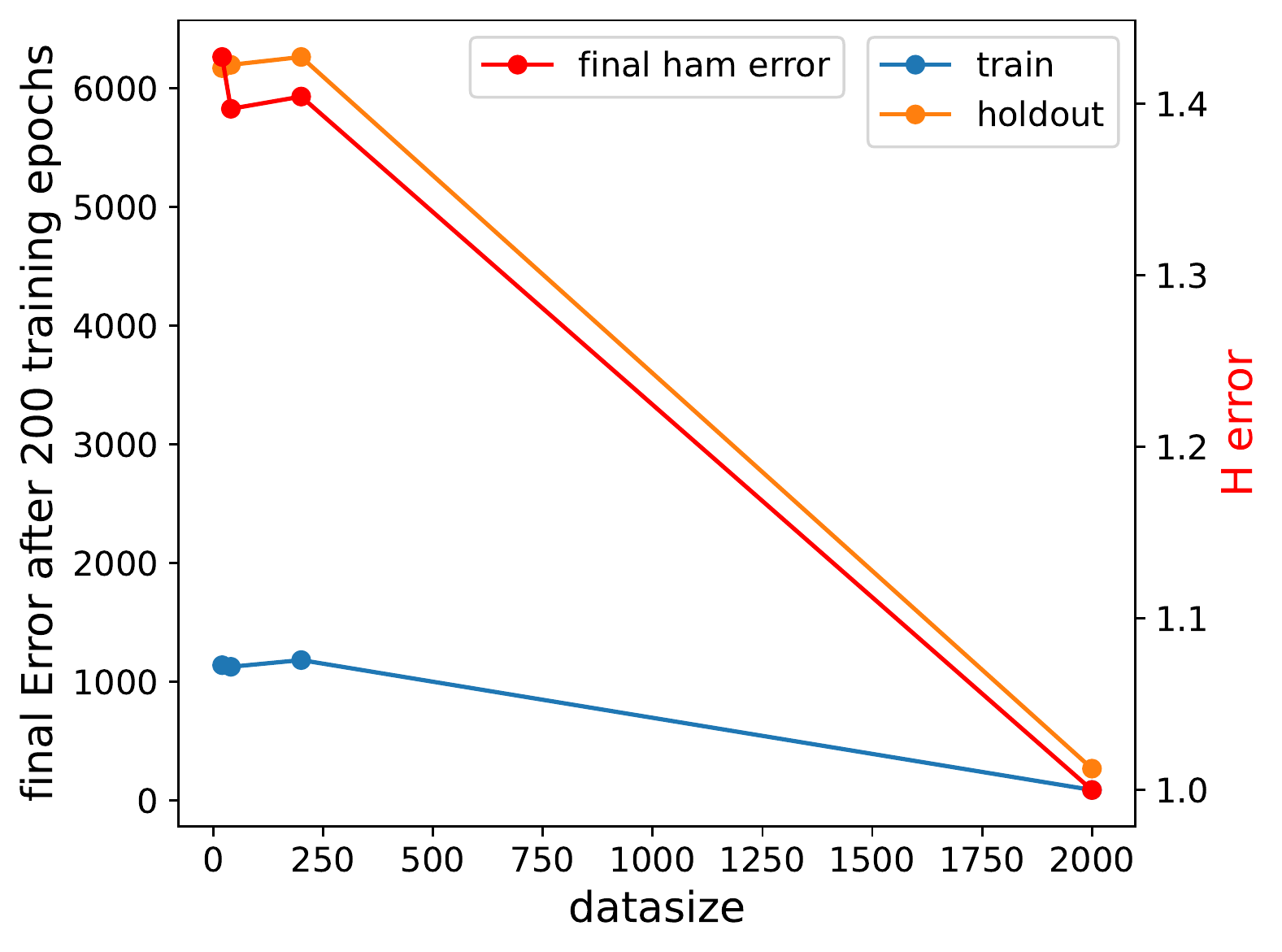}
  \caption{Effect of training data size on model generalization metrics: Hamiltonian error, unitary training $L_\text{model}(\mathcal{D}_\text{train})$ and validation (holdout) loss $L_\text{model}(\mathcal{D}_\text{val})$ for noisy single shot measurement-based unitary control of the transmon.}
  \label{fig:final_errors_datascaling_exp}
\end{figure}

A hallmark for a good ansatz for the model $\model$ estimating the dynamics of the controllable system would be less demand of supervised learning MDP data needed for low prediction error.

We consider the Hamiltonian error, unitary train and holdout error. Hamiltonian error $\delta$ is the spectral norm error between the learned and true system Hamiltonian. The others are mean squared errors. Cross-validation is used to estimate the model's generalization ability on a holdout dataset of unseen random unitary data, also sampled from the MDP transitions and collected by the policy $\pi$ during training.

As seen from Fig.~\ref{fig:datascaling_exp}, for the two-qubit transmon control problem, for very small dataset sizes comprising $20-200$ unitary transitions, the single step unitary prediction error is large compared to training with about $2,000$ unitaries or about $100$ full length pulses with $20$ timesteps, though the decrease in error is diminishing with dataset size. All errors are in agreement across the datasets over $200$ training epochs. This is further corroborated by Fig.~\ref{fig:final_errors_datascaling_exp} where the final errors after $200$ epochs are plotted. There is a reduction in the final errors for the $2000$ dataset size, but the improvement is diminishing in magnitude and plateaus at this loss for larger dataset sizes. This is still much less than what was required to train a neural network model for $\model$ during the initial stages of our research where the training dataset size needed to be of the order of $10^6$. Moreover, these experiments provide us with an idea of what dataset size to use to train the model $\model$ by setting the number of initial exploration MDP transitions to add to the policy's buffer for the transmon control problem. We also adopted multiple training phases to continuously train $\model$ using fresh batches of training data collected by the policy.

\section{{Leveraging the Learned Hamiltonian for the Two-qubit NV Center}}\label{app:graped_sac_nv}

{Similar to the results found in Sec.~\ref{ssec:graped_sac}, here we report the structural differences between the learned and target Hamiltonians for the two-qubit NV center.}

{The matrix difference between the true $H_0$ and learned Hamiltonian $H_0(\pmb{\zeta})$ is,
% we now add the matrix in latex below
\begin{widetext}
\begin{equation*}
    H-H_0(\pmb{\zeta}) = \left[\begin{matrix}
                    0.0116 & 0.0013i & -0.0001-0.0002i & -0.0007 \\
                        -0.0013i & -0.0111 & -0.0001+0.0002i & 0.0003+0.0003i \\
                        -0.0001+0.0002i & 0.0001+0.0002i & -0.0108 & -0.0005-0.0002i \\
                        -0.0007 & 0.0003-0.0003i & -0.0005+0.0002i  & -0.013 \\
                \end{matrix}\right]
\end{equation*}
\end{widetext}
Moreover, the non-linear relationship between the model prediction errors and the spectral norm error $\delta$ or the mean squared Pauli expectation value error is confirmed as before in Fig.~\ref{fig:local_hams_and_graped_conts_nv}(a). Local and global trajectory differences under a random control pulse and the results of using GRAPE on RL controllers are shown in Fig.~\ref{fig:local_hams_and_graped_conts_nv}(b) and (c) respectively. The learned Hamiltonian is able to improve the controller fidelities to greater than $0.999$.}

\begin{figure*}
  \centering
  \includegraphics[width=2.\columnwidth]{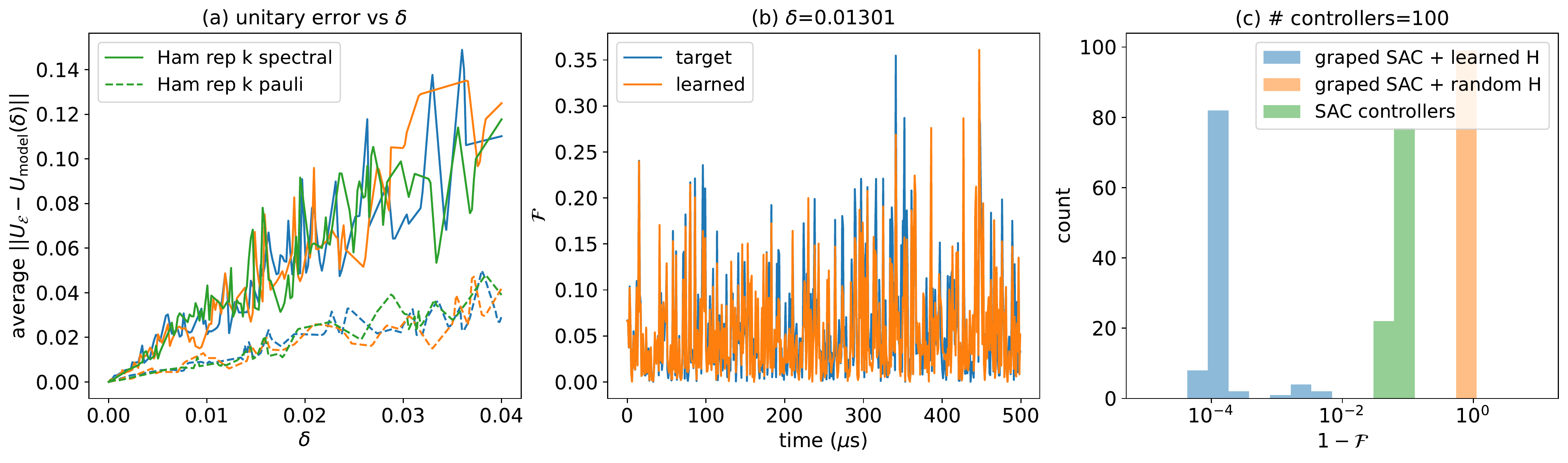}
  \caption{(a) The non-linear relationship between the prediction error $\left\|U_{\mathcal{E}} - U_{\model}\right\|$ and Hamiltonian spectral norm error or mean squared Pauli expectation value error $\delta$ for the two-qubit NV center Hamiltonian. For the same $1000$ random control pulses, we evaluate the average unitary prediction error of $\model$ with increasing $\delta$ for three different uniform randomly sampled two-qubit Hamiltonians $H_0(\pmb{\zeta})$. (b) Local and global unitary trajectories: $\mathcal{F}$ as a function of a random control pulse with either the learned $H_0(\pmb{\zeta})$ or true $H_0$. The learned trajectories and global trajectory overlap less with increasing time with the spectral norm error of $\delta=0.01301$ and a global phase factor $\Tr[ H-H_0(\pmb{\zeta})]$ of $\sim0.01$. (c) The learned $H_0(\pmb{\zeta})$ can be leveraged using GRAPE to further optimize the fidelities of LH-MBSAC's controllers. Repeating the procedure in Sec.~\ref{ssec:graped_sac}, yields fidelities of greater than $0.999$.}
  \label{fig:local_hams_and_graped_conts_nv}
\end{figure*}

{
\section{Three-qubit transmon Control Problem}\label{app:three-qubit-transmon}
% control Hamiltonian part of the 3-qubit transmon Hamiltonian

In this section we discuss the issue of scalability of LH-MBSAC's performance related to the three-qubit transmon control problem in Sec.~\ref{ssec:lims_and_silvers} in detail.

Working with two level systems, we extend the two-qubit transmon Hamiltonian to its three-qubit version $H_{\text{tra}}^{(3)}$. The system part generalizes trivially. For the control part $H_{\text{tra}_c}^{(3)}$,
we generalize the cross resonance interaction presented in Ref.~\cite{threequbitmodel} to construct the following time-dependent part of the three-qubit transmon Hamiltonian,
\begin{multline}\label{eq:three-qubit-hamiltonian}
  \frac{H_{\text{tra}_c}^{(3}(t)}{\hbar} = \sum_{l=1}^3 \Big(a_l(t)(Z_{l} X_{l+1} + X_{l+1} + Y_{l+1} + Z_{l})\\
  + b_l(t)(X_lZ_{l+1} + X_{l} + Y_{l} + Z_{l+1})\Big)
\end{multline}
where $a_l(t), b_l(t)$ are the real drive amplitudes and $X_l, Y_l, Z_l$ are the corresponding Pauli operators on the $l$th qubit.

To start, we mention our hyperparameter strategy. Only an initial hyperparameter search is performed for the two-qubit transmon control problem, and we were successfully able to transfer the same hyperparameters to all problems in the paper that were studied including the ones presented in Fig.~\ref{fig:sample_complexity_results}.

It is a desirable property for the stabiltiy of RL algorithms to be robust to hyperparameter changes for different target problems, which we found to be the case. The search was only conducted for the model-free SAC since LH-MBSAC is just a model-based augmentation of the underlying SAC algorithm so there is no strong reason for the hyperparameters to fail to transfer.

However, for the three-qubit transmon control problem, we encountered issues and had to repeat the search. This was extensive, and what we focused on are: more initial exploration data, using bigger layer sizes for the policy and value function neural networks, changing the learning and update rates for the policy and value functions, amongst other things. An extremely thorough search is difficult since the problem is more computationally challenging, and it is hard to determine when to terminate the training during a trial run that necessarily needs to be premature during the hyperparameter search. Please see the accompanying code for the list of hyperparameters we searched over using Bayesian optimization in \texttt{tune\_hypers.py} along with some results in the \texttt{hyper\_tests} folder.

\begin{figure}
  \centering
  \includegraphics[width=0.98\linewidth]{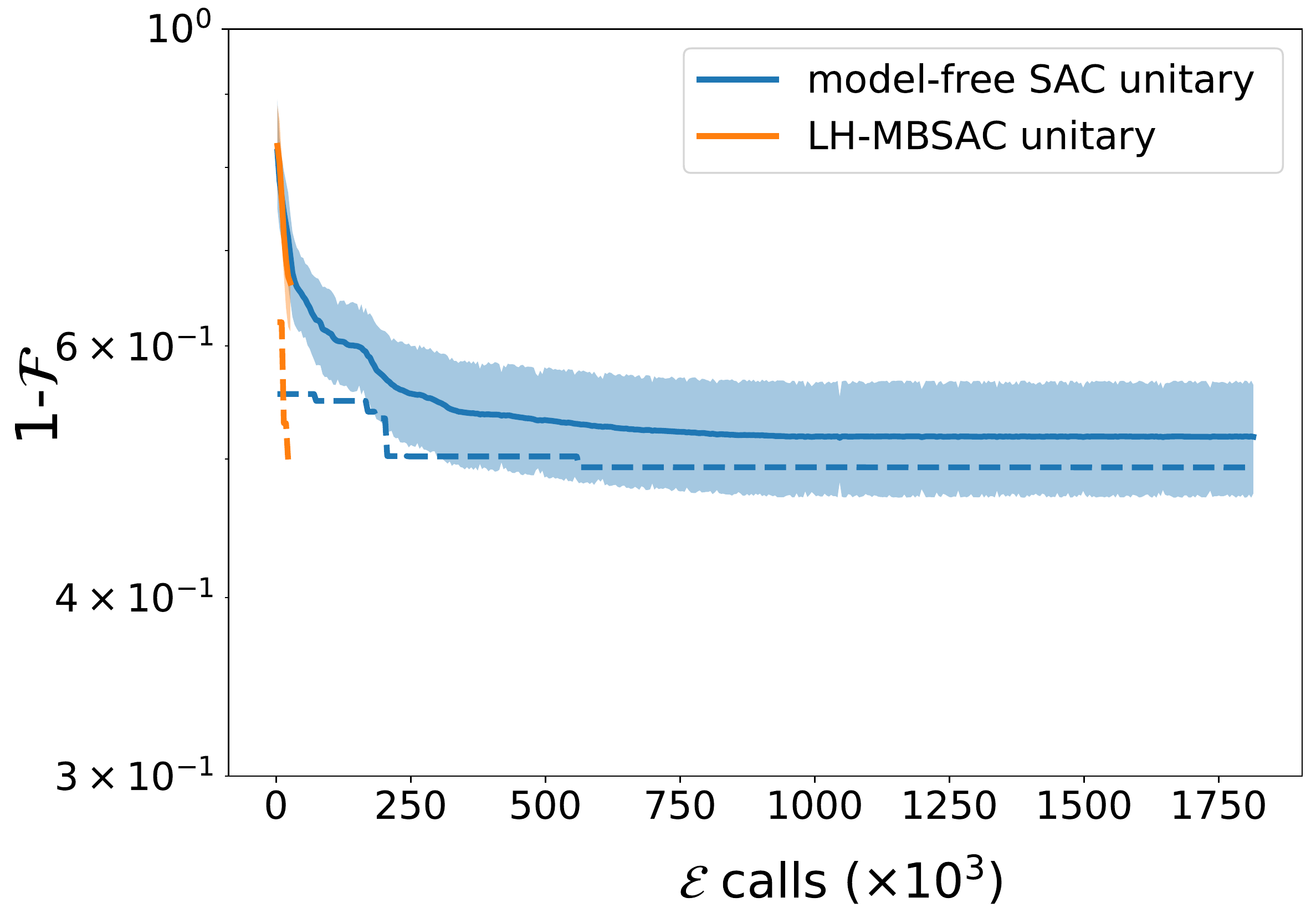}
  \caption{Noiseless unitary sample complexity for the three-qubit transmon where the target gate is the Toffoli gate. Since LH-MBSAC is based on SAC, the latter's training curves are obtained first to see if it viably solves the problem, and it was trained for much longer i.e.
  in the order of millions of samples as seen in Fig.~\ref{fig:three-qubit-transmon-sample-complexity}. Mean (solid) and maximum fidelities (dashed) saturate as the policy and value functions gradients and outputs saturate due to the agent getting stuck in a suboptimal extremum of the optimization landscape.}
  \label{fig:three-qubit-transmon-sample-complexity}
\end{figure}

Furthermore, we make observations that make this issue seem less like a hyperparameter issue and more like an optimization landscape problem:
\begin{enumerate}
\item The values and the gradients for policy and value functions that saturate are both stuck in suboptimal extrema and ultimately we get stuck at a prematurely optimized reward function. This is illustrated in Fig.~\ref{fig:three-qubit-transmon-sample-complexity}. Essentially, SAC gets stuck in a loop mining medium level fidelities and its policy outputs saturate on the extremes of the control amplitudes. It is already detailed in Sec.~\ref{ssec:lims_and_silvers} that RL pulses are biased towards maintaining high intermediate fidelities due to the nature of the MDP used in the paper. Fig.~\ref{fig:rl_vs_grapecontrollers} example pulses found by RL vs.\ GRAPE for the two-qubit transmon, confirming this.
\item Since we have the model Hamiltonian, we insert it into GRAPE initialized with the highest fidelity SAC controller values, and it also gets stuck (at slightly better fidelities).
\end{enumerate}
Despite these issues, the system Hamiltonian is still learned. It can be inserted into GRAPE with uniform random initialization of control pulse parameters to achieve fidelities of over $0.999$.
}

\section{Comparison of Fidelities for Lindbladian Dynamics}\label{app:fid_metrics_comparison}

We study the agreement between three different fidelity measures of realized noisy gates on open systems with Lindblad decay and decoherence for the two-qubit transmon gate control problem. The fidelity measures are the diamond norm fidelity~\cite{diamondnorm}, the generalized state fidelity~\cite{flammia_direct_fidelity_2011}, and the average gate fidelity~\cite{Uhlmann}. The diamond norm fidelity, derived from the diamond norm or the completely bounded trace norm, is the most expensive to compute as it involves solving a convex optimization problem:
\begin{multline}
  \mathcal{F}_\diamond(\choi(\mathbf{u}(t), t), \choi_\text{target})
  = 1-\|\choi(\mathbf{u}(t), t)- \choi_\text{target}\|_\diamond \\
  = 1-\max_{\rho}\left\|\choi(\mathbf{u}(t), t) \circ \rho - \choi_\text{target} \circ \rho \right\|_1,
  \label{eq:diamond_norm_under_the_hood}
\end{multline}
where the maximization is over the space of all density matrices $\rho$. This can be done by solving an equivalent semi-definite program~\cite{watrous_diamond_norm_compute}. $0.5 \leq \mathcal{F}_\diamond(\choi(\mathbf{u}(t), t) \leq 1$.

To study the sensitivities of the measures to dissipation and their agreement w.r.t. each other, we consider low, medium and high dissipation regimes. We evaluate $100$ of our controllers found for the noisy single shot measurements setting of the two-qubit transmon in these regimes. The results are plotted in Fig.~\ref{fig:fid_comparison}. Here, \texttt{deca} and \texttt{deco} refer to inverse decoherence and decay rates $2/T^*_l, 2/T_l$ respectively, for the $l$th qubit, measured in MHz. We re-normalize the trace of the realized operator $\choi(\mathbf{u}(t), t)$ during our experiments, as is standard practice. Due to the exhaustive nature of its computation, $\mathcal{F}_\diamond$ is the most sensitive to noise and loss of coherence out of all the measures. The generalized state fidelity is the least sensitive and the average gate fidelity falls in the middle. For very low to medium dissipation levels, e.g., $(0.05, 0.05)$, $(0.05, 0.1)$, or $(0.05, 0.2)$ for the pair $(\texttt{deca}, \texttt{deco})$, the generalized state fidelity is near perfect while the gate and diamond norm fidelities are more sensitive and closer to $0.9$. For this reason, in Sec.~\ref{ssec:open_sys_results}, we chose to use the diamond norm fidelity to more accurately gauge controller performance---this was especially true for the low dissipation regime results.

As a side note, some controllers shown in Fig.~\ref{fig:fid_comparison} are more robust to dissipation than others as revealed by the noisy variation across the controller index vs. fidelity plot. The controllers are not ordered, so the fidelity in the zero dissipation regime has some noise/variation as seen for $\texttt{deca}, \texttt{deco} = (0.05, 0.05)$. Across all the subfigures, the robustness is captured by all the fidelity measures where the variation magnitudes and positions are more or less aligned.

\begin{figure*}
  \centering
  \includegraphics[width=0.98\linewidth]{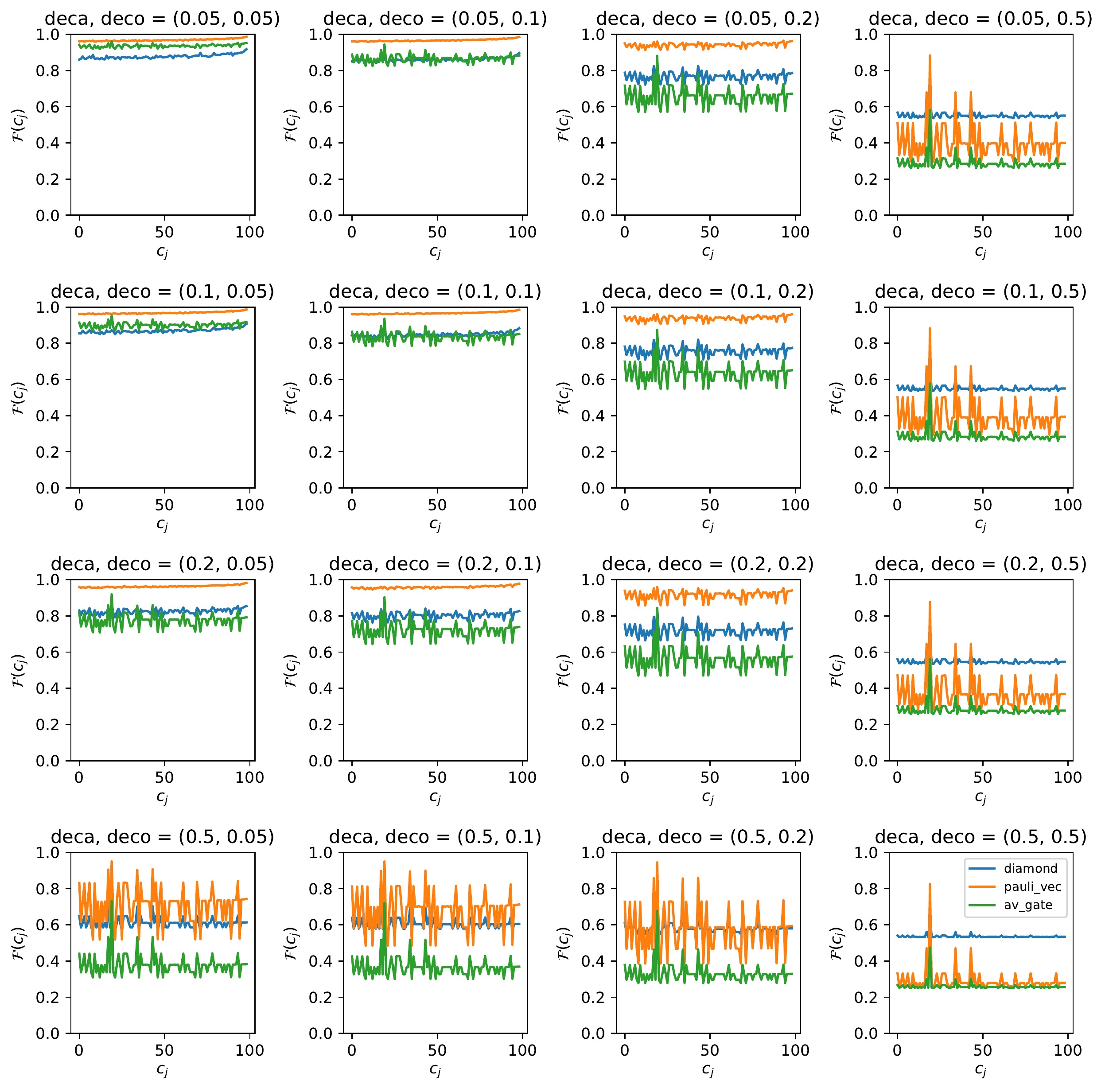}
  \caption{How much the fidelity measures relate to one another as the dissipation strength varies in terms of the decoherence and the decay coefficients in Eq.~\eqref{eq:Lindblad_operators} for the Lindbladian $l_d$ operators. Here, \texttt{deca}, \texttt{deco} refer to inverse decay and decoherence rates $2/T^*_l, 2/T_l$ respectively, for the $l$th qubit measured in MHz. The x-axis refers to a controller $c_j$ obtained for the two-qubit transmon gate control problem with single shot measurement noise where the target is the CNOT gate. The controllers are in random order w.r.t. the fidelity, but the ordering is preserved across each subfigure. The number of single shot measurements is $10^6$ and \texttt{diamond}, \texttt{pauli\_vec}, \texttt{av\_gate} refer to the diamond norm fidelity~\cite{diamondnorm}, the generalised state fidelity~\cite{flammia_direct_fidelity_2011} and the average gate fidelity~\cite{Uhlmann}.
  }
  \label{fig:fid_comparison}
\end{figure*}

\end{document}

%% file: preamble.tex
\usepackage{float}
\makeatletter
\let\newfloat\newfloat@ltx
\makeatother
\usepackage[english]{babel}
\usepackage[utf8]{inputenc}
\usepackage{graphics}
\usepackage{selinput}
\usepackage[normalem]{ulem}
\usepackage[shortlabels]{enumitem}
\usepackage{braket}
\usepackage{amsthm}
\usepackage{mathtools}
\usepackage{physics}
\usepackage{xcolor}
\usepackage{graphicx}
\usepackage{adjustbox}
\usepackage{placeins}
\usepackage[T1]{fontenc}
\usepackage{lipsum}
\usepackage{csquotes}
\usepackage{bm}
\usepackage{siunitx}
\usepackage[caption=false]{subfig}
\usepackage[linesnumbered,ruled,noend]{algorithm2e}
\SetKwInput{kwInit}{Init}

\usepackage[makeroom]{cancel}
\usepackage[toc,page]{appendix}
\usepackage[colorlinks=true,citecolor=blue,linkcolor=magenta]{hyperref}

\usepackage{tikz}
\usetikzlibrary{arrows}
\usetikzlibrary{automata}
\usetikzlibrary{matrix}
\usetikzlibrary{calc}
\usetikzlibrary{positioning}
\usetikzlibrary{graphs}
\usetikzlibrary{shapes.geometric}
\usetikzlibrary{backgrounds}

\usepackage{pgffor}

\newcommand*{\SU}[1]{\mathbb{s}\mathbb{u}(#1)}

\newcommand{\choi}{\mathbf{\Phi}}

\renewcommand{\geq}{\geqslant}
\renewcommand{\leq}{\leqslant}

\DeclareMathOperator*{\argmax}{arg\,max}

\DeclareMathOperator*{\model}{\mathbf{M}_{\pmb{\zeta}}}
\newcommand{\bs}{\textsf{BS}}

\def\be{\begin{equation}}
\def\ee{\end{equation}}
\def\bs{\begin{split}}
\def\e{\end{split}}
\def\ba{\begin{eqnarray}}
\def\bea{\begin{eqnarray}}

\def\tea{\end{eqnarray}}
\def\ea{\end{eqnarray}}
\def\eea{\end{eqnarray}}

\def\eye{\mathds{1}}

\def\su{\mathfrak{s}\mathfrak{u}}
\def\H{\tilde{H}}

\def\su{\mathfrak{s}\mathfrak{u}}
\def\SU{\text{SU}}

\DeclareMathOperator{\s}{\mathbf{s}}
\DeclareMathOperator{\rew}{r}
\DeclareMathOperator{\action}{\mathbf{a}}
\DeclareMathOperator{\nntheta}{\mathbf{\theta}}

\newtheorem{theorem}{Theorem}
\newtheorem{lemma}{Lemma}
\newtheorem{proposition}{Proposition}

\newenvironment{specialproof}{\textit{Proof:}}{\hfill$\square$}
\usepackage{amssymb}
\usepackage{dsfont}